\documentclass[nohyperref]{article}

\usepackage{microtype}
\usepackage{graphicx}
\usepackage{subcaption}
\usepackage{booktabs} 
\usepackage{pbox}

\usepackage{multirow} 
\usepackage{bbm}
\usepackage{array}
\usepackage{adjustbox}
\usepackage{mathtools}
\usepackage{hyperref}   
\usepackage{xcolor}
\usepackage{bbding}
\usepackage{makecell}
\usepackage{algorithm}
\usepackage{algorithmic}
\usepackage[algo2e,ruled,linesnumbered,vlined,boxed,commentsnumbered]{algorithm2e}
\usepackage{titlesec}
\usepackage{colortbl}
\usepackage{amsthm,amssymb,amsfonts,amsmath}

\usepackage{amsmath,amsfonts,bm}

\def\eqref#1{equation~\ref{#1}}

\def\1{\bm{1}}

\def\vo{{\bm{o}}}

\DeclareMathAlphabet{\mathsfit}{\encodingdefault}{\sfdefault}{m}{sl}
\SetMathAlphabet{\mathsfit}{bold}{\encodingdefault}{\sfdefault}{bx}{n}

\DeclareMathOperator*{\argmax}{arg\,max}
\DeclareMathOperator*{\argmin}{arg\,min}

\usepackage{amsthm}

\newtheorem{definition}{Definition}

\newtheorem{proposition}{Proposition}

\usepackage{mathtools}
\newcommand{\defeq}{\vcentcolon=}

\usepackage{hyperref}

\renewcommand{\defeq}{\coloneqq}
\def\algcomment#1{\textcolor[rgb]{0,0,0.6}{\# #1}}
\def\down#1{\textcolor[rgb]{0,0.6,0}{$\downarrow #1 \%$}}

\newlength{\oldtextfloatsep}
\def\Y{{\mathcal{Y}}}

\def\ovo{{\overline{V}}}
\def\vo{{V}}
\definecolor{lightred}{RGB}{255,153,51}
\definecolor{lightblue}{RGB}{51,153,255}
\definecolor{lightgreen}{RGB}{144,238,144}
\def\checkmark{\tikz\fill[scale=0.4](0,.35) -- (.25,0) -- (1,.7) -- (.25,.15) -- cycle;}

\usepackage[accepted]{icml2022}

\usepackage{amsmath}
\usepackage{amssymb}
\usepackage{mathtools}
\usepackage{amsthm}
\usepackage{threeparttable}
\usepackage[capitalize,noabbrev]{cleveref}
\theoremstyle{plain}

\usepackage[textsize=tiny]{todonotes}

\begin{document}

\twocolumn[
\icmltitle{On Collective Robustness of Bagging Against Data Poisoning}



\icmlsetsymbol{equal}{*}

\begin{icmlauthorlist}
\icmlauthor{Ruoxin Chen}{sjtu}
\icmlauthor{Zenan Li}{sjtu}
\icmlauthor{Jie Li}{sjtu}
\icmlauthor{Chentao Wu}{sjtu}
\icmlauthor{Junchi Yan}{sjtu}
\end{icmlauthorlist}

\icmlaffiliation{sjtu}{Department of Computer Science and Engineering and MoE Key Lab of Artificial Intelligence, Shanghai Jiao Tong University, Shanghai, China. Jie Li and Junchi Yan are also with Shanghai AI Laboratory, Shanghai, China}
\icmlcorrespondingauthor{Jie Li}{\quad lijiecs@sjtu.edu.cn}

\icmlkeywords{Machine Learning, ICML}

\vskip 0.3in
]
\printAffiliationsAndNotice{} 



\begin{abstract}
Bootstrap aggregating (bagging) is an effective ensemble  protocol, which is believed can enhance robustness by its majority voting mechanism. Recent works further prove the sample-wise robustness certificates for certain forms of bagging (e.g. partition aggregation). Beyond these particular forms, in this paper, \emph{we propose the first collective certification for general bagging to compute the tight robustness against the global poisoning attack}. Specifically, we compute the maximum number of simultaneously changed predictions via solving a binary integer linear programming (BILP) problem. Then we analyze the robustness of vanilla bagging and give the upper bound of the tolerable poison budget. Based on this analysis, \emph{we propose hash bagging} to improve the robustness of vanilla bagging almost for free. This is achieved by modifying the random subsampling in vanilla bagging to a hash-based deterministic subsampling, as a way of controlling the influence scope for each poisoning sample universally. Our extensive experiments show the notable advantage in terms of applicability and robustness. Our code is available at \url{https://github.com/Emiyalzn/ICML22-CRB}.
\end{abstract}

\section{Introduction}
Bagging~\cite{breiman1996bagging}, refers to an ensemble learning protocol that \emph{trains sub-classifiers on the subsampled sub-trainsets and makes predictions by majority voting}, which is a commonly used method to avoid overfitting. Recent works~\cite{biggio2011bagging,levine2021deep, jia2021intrinsic} show its superior certified robustness in defending data poisoning attacks. Moreover, compared to other certified defenses, bagging is a natural plug-and-play method with a high compatibility with various model architectures and training algorithms, which suggests its great potential.

Some works~\cite{levine2021deep, jia2021intrinsic,wang2022improved} have proved the sample-wise robustness certificates against the sample-wise attack (the attacker aims to corrupt the prediction for the target data) for certain forms of bagging. However, we notice that, \emph{there is a white space in the collective robustness certificates against the global poisoning attack} (the attacker attempts to maximize the number of simultaneously changed predictions when predicting the testset), although the global attack is more general and critical than the sample-wise attack for: I) the sample-wise attack is only a variant of the global poisoning attack when the testset size is one; II) unlike adversarial examples~\cite{RN6} which is sample-wise, data poisoning attacks are naturally global, where the poisoned trainset has a global influence on all the predictions; III) the global attack is believed more harmful than the sample-wise attack. Current works~\cite{levine2021deep, jia2021intrinsic} simply \emph{count the number of robust predictions guaranteed by the sample-wise certification}, as a lower bound of the collective robustness. However, this lower bound often overly under-estimates the actual value. We aim to provide a formal collective certification for general bagging, to fill the gap in analyzing the certified robustness of bagging.

In this paper, \emph{we take the first step towards the collective certification for general bagging.} Our idea is to formulate a binary integer linear programming (BILP) problem, of which objective function is to maximize the number of simultaneously changed predictions w.r.t. the given poison budget. The certified collective robustness equals the testset size minus the computed objective value. To reduce the cost of solving the BILP problem, a decomposition strategy is devised, which allows us to compute a collective robustness lower bound within a linear time of testset size.

Moreover, we analyze the certified robustness of vanilla bagging, demonstrating that it is not an ideal certified defense by deriving the upper bound of its tolerable poison budget. To address this issue, \emph{we propose hash bagging to improve the robustness of vanilla bagging almost for free.} Specifically, we modify the random subsampling in vanilla bagging to hash-based subsampling, to restrict the influence scope of each training sample within a bounded number of sub-trainsets deterministically. We compare hash bagging to vanilla bagging to show its superior certified robustness and the comparable accuracy. Furthermore, compared to prior elaborately designed bagging-based defenses~\cite{levine2021deep,jia2021intrinsic}, hash bagging is a more general and practical defense method, which covers almost all forms of bagging. 
\textbf{The main contributions are:}

1) For the first time to our best knowledge, we derive the collective certification for general bagging. We accelerate the solving process by decomposition. Remarkably, our computed certified collective robustness is theoretically better than that of the sample-wise certifications.

2) We derive an upper bound of tolerable poison budget for bagging. Our derived bound is tight if we only have access to the sub-trainsets and sub-classifier predictions.

3) We propose \emph{hash bagging} as a defense technique to improve the robustness for vanilla bagging almost for free, in the sense of neither introducing additional constraints on the hyper-parameters nor restricting the forms of bagging.

4) We evaluate our two techniques empirically and quantitatively on four datasets: collective certification and hash bagging. Results show: i) collective certification can yield a much stronger robustness certificate. ii) Hash bagging effectively improves vanilla bagging on the certified robustness.

\vspace{-10pt}
\section{Related Works}
\vspace{-5pt}

Both machine-learning classifiers (e.g. Bayes and SVM) and neural-network classifiers are vulnerable to data poisoning ~\cite{li2020backdoor,li2022few,nelson2008exploiting,biggio2012poisoning,xiao2015is,yao2019latent,zhang2020black,liu2019abs}. Since most heuristic defenses~\cite{chen2019deepinspect,gao2019strip,tran2018spectral,liu2019abs,qiao2019defending} have been broken by the new attacks~\cite{koh2018stronger,tramer2020on}, developing certified defenses is critical.

\vspace{-2pt}

\textbf{Certified defenses against data poisoning.}
Certified defenses~\cite{steinhardt2017certified,Wang2020OnCR} include random flipping~\cite{Rosenfeld2019CertifiedRT}, randomized smoothing~\cite{weber2020rab}, differential privacy~\cite{ma2019data} and bagging-based defenses~\cite{levine2021deep,jia2021intrinsic}. Currently, only the defenses~\cite{ma2019data,jia2022rnn,jia2021intrinsic,levine2021deep} are designed for the general data poisoning attack (the attacker can arbitrarily insert/delete/modify a bounded number of samples). However, their practicalities suffer from various limitations. \cite{ma2019data} is limited to the training algorithms with the differential privacy guarantee. ~\cite{jia2022rnn} certify the robustness for the machine-learning classifiers kNN/rNN (Nearest Neighbors), which might be unable to scale to the large tasks. Currently, only two bagging variants~\cite{jia2021intrinsic,levine2021deep} have demonstrated the high compatibility w.r.t. the model architecture and the training algorithm, with the state-of-the-art certified robustness. Their success highlights the potential of bagging, which motivates us to study the robustness for general bagging.

\vspace{-2pt}

\textbf{Robustness certifications against data poisoning.}
Current robustness certifications~\cite{Wang2020OnCR,ma2019data,jia2021intrinsic,jia2022rnn,levine2021deep} against data poisoning are mainly focusing on the sample-wise robustness, which evaluates the robustness against the sample-wise attack. However, the collective robustness certificates are rarely studied, which might be a more practical metric because the poisoning attack naturally is a kind of global attack that can affect all the predictions. To our best knowledge, only~\cite{jia2022rnn} considers the collective robustness against global poisoning attack. Specifically, it gives the collective certification for a machine-learning classifier rNN, but the certification is based on the unique geometric property of rNN.\footnote{~\cite{schuchardt2021collective} derive the collective certificates for GNN. Their collective certificates are focusing on the adversarial examples, instead of data poisoning.}

\begin{table}[tb!]
\centering
\caption{Notations.}
\label{tab:notation}
    \begin{adjustbox}{width=0.98\linewidth}
    \begin{tabular}{ c | c }
    \toprule
     Notation & Description \\
     \hline
     $K$ & The sub-trainset size.\\
     \hline
     $G$ & The number of sub-trainsets.\\
     \hline
     $N$ & The trainset size.\\
     \hline
     $\mathcal{D}_{train}=\{s_i\}_{i=0}^{N-1}$ & The trainset consisting of $N$ training samples $\{s_i\}_{i=0}^{N-1}$.\\
     \hline
     $\mathcal{D}_{test}=\{x_j\}_{j=0}^{M-1}$ & The trainset consisting of $M$ testing samples $\{x_j\}_{j=0}^{M-1}$.\\
     \hline
     $y \in \Y$ & $y$ and $\Y$  denote the class and the output space respectively.\\
     \hline
     $\mathcal{D}_{g}$ & The $g$-th sub-trainset.\\
     \hline
     $f_{g}(\cdot)$ & The $g$-th sub-classifier in bagging.\\
     \hline
     $g(\cdot)$ & The ensemble classifier consisting of all the sub-classifiers.\\
     \hline
     $\vo_{x}(y)$ & The number of votes for the class $y \in \Y$ when predicting $x$.\\
     \hline
     $\rm{Hash}(\alpha)$ & The hash value of $\alpha$.\\
    \bottomrule
    \end{tabular}   
    \end{adjustbox}
\end{table}

\vspace{-10pt}
\section{Collective Certification to Bagging}
\vspace{-5pt}
In this section, first we formally define vanilla bagging and the threat model, as the basement of the collective certification. Then we propose the collective certification, and analyze the upper bound of the tolerable poison budget. All our
notations are summarized in Table ~\ref{tab:notation}.

\begin{definition}[Vanilla bagging]
Given a trainset $\mathcal{D}_{train}=\{s_i\}_{i=0}^{N-1}$ where $s_i$ refers to the $i$-th training sample, following ~\cite{breiman1996bagging,jia2021intrinsic, levine2021deep}, vanilla bagging can be summarized into three steps: \\
i) Subsampling: construct $G$ sub-trainsets $\mathcal{D}_g$ (of size $K$) ($g=0,\ldots,G-1$), by subsampling $K$ training samples from $\mathcal{D}_{train}$ $G$ times;\\
ii) Training: train the $g$-th sub-classifier $f_g(\cdot)$ on the sub-trainset $\mathcal{D}_g$  ($g=0,\ldots,G-1$);\\ 
iii) Prediction: the ensemble classifier (denoted by $g(x)$) makes the predictions, as follow: 
\begin{equation}
\footnotesize
    g(x) = \argmin\limits_{y} \argmax\limits_{y \in \Y} \vo_x(y) 
\end{equation}
where $\vo_x(y) \defeq \sum_{g=0}^{G-1} \mathbb{I} \{f_g(x)=y\}.$ ($\mathbb{I}\{\}$ is the indicator function) is the number of sub-classifiers that predict class $y$. $\argmin_{y}$ means that, $g(x)$ predicts \textbf{the majority class of the smallest index} if there exist multiple majority classes.
\end{definition}

\vspace{-10pt}
\subsection{Threat Model}\label{sec:threat} 
\vspace{-5pt}
We assume that \emph{the sub-classifiers are extremely vulnerable to the changes in their sub-trainsets}, since our certification is agnostic towards the sub-classifier architecture. In another word, the attacker is considered to fully control the sub-classifier $f_g$ once the sub-trainset $\mathcal{D}_g$ is changed.\\
\textbf{Attacker capability:} the attacker is allowed to insert $r_{\rm{ins}}$ samples, delete  $r_{\rm{del}}$ samples, and modify $r_{\rm{mod}}$ samples. \\
\textbf{Attacker objective:} for the \emph{sample-wise attack} (corresponding to the sample-wise certification), the attacker aims to change the prediction for the target data. For the \emph{global poisoning attack} (corresponding to the collective certification), the attacker aims to maximize the number of simultaneously changed predictions when predicting the testset.


\vspace{-10pt}
\subsection{$\mathbf{(P1)}$: Collective Certification of Vanilla Bagging}
Given the sub-trainsets and class distribution of each testing sample, we can compute the collective robustness for vanilla bagging, as shown in Prop.~\ref{pro:bagging}.
\begin{proposition}[Certified collective robustness of vanilla bagging]
\label{pro:bagging}
For testset $\mathcal{D}_{test}=\{x_j\}_{j=0}^{M-1}$, we denote $\hat{y}_j= g(x_j)$ ($j=0,\ldots,M-1$) the original ensemble prediction, and $\mathcal{S}_i=\{g \mid s_i \in \mathcal{D}_g\}$ the set of the indices of the sub-trainsets that contain $s_i$ (the $i$-th training sample). Then, the maximum number of simultaneously changed predictions (denoted by $M_{\rm{ATK}}$) under $r_{\rm{mod}}$ adversarial modifications, is computed by $\mathbf{(P1)}$:
{
\small
\begin{align} 
    &\mathbf{(P1):} \quad M_{\rm{ATK}}= \max_{P_0,\ldots,P_{N-1}} \; \sum_{x_j \in \mathcal{D}_{test}} \mathbb{I} \left\{ \ovo_{x_j}(\hat{y}_j) < \right. \nonumber\\ 
    & \left. \qquad \qquad \qquad \qquad \max_{y \neq \hat{y}_j} \left[\ovo_{x_j}(y)+ \frac{1}{2}\mathbb{I}\{y <\hat{y}_j\}\right] \right\} \label{eq:obj}\\
     & s.t. \quad [P_0, P_1, \ldots, P_{N-1}] \in \{0,1\}^N \label{eq:c0}\\
     & \qquad \sum_{i=0}^{N-1} P_i \leq r_{\rm{mod}} \label{eq:c1}\\
     & \ovo_{x_j}(\hat{y}_j) = \underbrace{\vo_{x_j}(\hat{y}_j)}_{\textbf{Original votes}}- \underbrace{\sum_{g=0}^{G-1} \mathbb{I}\{g \in \bigcup_{\forall i,P_i=1}\mathcal{S}_i\} \mathbb{I}\{ f_g(x_j)=\hat{y}_j \}  }_{\textbf{Influenced votes}} \nonumber \\
     & \qquad  \qquad \forall x_j \in \mathcal{D}_{test}, \; \hat{y}_j=g(x_j) \; \label{eq:c2}\\
     & \ovo_{x_j}(y)= \underbrace{\vo_{x_j}(y)}_{\textbf{Original votes}}+ \underbrace{\sum_{g=0}^{G-1} \mathbb{I}\{ g \in \bigcup_{\forall i, P_i=1}\mathcal{S}_i\} \mathbb{I}\{ f_g(x_j) \neq y \} }_{\textbf{Influenced votes}} \nonumber \\
     & \qquad \qquad \forall x_j \in \mathcal{D}_{test},\; \forall y \in \Y, y\neq \hat{y}_j \label{eq:c3} 
\end{align}
}
The certified collective robustness is $M-M_{\rm{ATK}}$.
\end{proposition}
We explain each equation. \textbf{Eq. (\ref{eq:obj}):} the objective is to maximize the number of simultaneously changed predictions. Note that a prediction is changed if there exists another class with more votes (or with the same number of votes but of the smaller index). \textbf{Eq. (\ref{eq:c0}):} $[P_0, \ldots, P_{N-1}]$ are the binary variables that represent the poisoning attack, where $P_i=1$ means that the attacker modifies $s_i$. \textbf{Eq. (\ref{eq:c1}):} the number of modifications is bounded within $r_{\rm{mod}}$. \textbf{Eq. (\ref{eq:c2}):} $\ovo_{x_j}(\hat{y}_j)$, the minimum number of votes for class $\hat{y}_j$ (after being attacked), equals to the original value minus the number of the influenced sub-classifiers whose original predictions are $\hat{y}_j$. \textbf{Eq. (\ref{eq:c3}):} $\ovo_{x_i}(y)$ ($y \neq y_i$), the maximum number of votes for class $y: y \neq \hat{y}_j$ (after being attacked), equals to the original value plus the number of influenced sub-classifiers whose original predictions are not $y$, because that, under our threat model, the attacker is allowed to arbitrarily manipulate the predictions of those influenced sub-classifiers.

\subsection{Remarks on Proposition~\ref{pro:bagging}}
We give our discussion and the remark marked with $*$ mean that the property is undesirable needing improvement.

    \textbf{1)} \textbf{Tightness.} The collective robustness certificates computed from $\mathbf{(P1)}$ is tight. \vspace{5pt}\\ 
    \textbf{2)}  \textbf{Sample-wise certificate.} We can compute the tight sample-wise certificate for the prediction on the target data $x_{\rm{target}}$, by simply setting $\mathcal{D}_{test}= \{x_{\rm{target}} \}$. \vspace{5pt}\\
    \textbf{3)}  \textbf{Certified accuracy.} We can compute \emph{certified accuracy} (the minimum number of correct predictions after being attacked) if given the oracle labels. Specifically, we compute the certified accuracy over the testset $\mathcal{D}_{test}$, simply by modifying $\sum_{x_j \in \mathcal{D}_{test}}$ in Eq. (\ref{eq:obj}) to $\sum_{x_j \in \Omega}$, where $\Omega$ is
  $\Omega=\{ x_j \in  \mathcal{D}_{test}: g(x_j) \text{ predicts correctly}\}$. The certified accuracy is $(|\Omega|- M_{\rm{ATK}})/M$ where $|\Omega|$ refers to the cardinality of the set $\Omega$. Actually, certified accuracy measures the worst accuracy under all the possible accuracy degradation attacks within the poison budget. Our computed certified accuracy is also tight. \vspace{5pt}\\
    \textbf{4)}  \textbf{Reproducibility requirement*.}  Both subsampling and training are required to be reproducible, because certified robustness is only meaningful for deterministic predictions. Otherwise, without the reproducibility, given the same trainset and testset, the predictions might be discrete random variables for the random operations in subsampling/training, such that we may observe two different predictions for the same input if we run the whole process (bagging and prediction) twice, even without being attacked. \vspace{5pt}\\
    \textbf{5)}  \textbf{NP-hardness*.} $\mathbf{(P1)}$ is NP-hard as it can be formulated as a BILP problem. We present more details in Appendix (Section~\ref{sec:NPhard}).

\subsection{Addressing NP-hardness by Decomposition}
\label{sec:decomposition}
Decomposition~\cite{quantumDecom,RAO20082768} allows us to compute a certified collective robustness \emph{lower bound} instead of the exact value. Specifically, we first split $\mathcal{D}_{test}$ into $\Delta$-size sub-testsets (denoted by $\mathcal{D}^\mu: \mu=0,\ldots,\lceil M/\Delta \rceil-1$). Here we require the size of the last sub-testset is allowed to be less than $\Delta$. Then we compute the maximum number of simultaneously changed predictions (denoted by $M^{\mu}_{\rm{ATK}}$) for each sub-testset $\mathcal{D}^\mu$ under the given poison budget. \textbf{We output $M-\sum_\mu M^{\mu}_{\rm{ATK}}$ as a collective robustness lower bound.} Remarkably, by decomposition, the time complexity is significantly reduced from an exponential time (w.r.t. $M$) to a linear time (w.r.t. $M$), as the time complexity of solving the $\Delta$-scale sub-problem can be regarded as a constant. Generally, $\Delta$ controls a trade-off between the certified collective robustness and the computation cost: \emph{as we consider the influence of the poisoning attack more holistically (larger $\Delta$), we can obtain a tighter lower bound at a cost of much larger computation.} In particular, our collective certification is degraded to be the sample-wise certification when $\Delta=1$.

\subsection{Upper Bound of Tolerable Poison Budget}
Based on Eq. (\ref{eq:c2}), Eq. (\ref{eq:c3}) in $\mathbf{(P1)}$, we can compute the upper bound of tolerable poison budget for vanilla bagging. 
\begin{proposition}[Upper bound of tolerable poison budget] \label{pro:upper}
Given $\mathcal{S}_i= \{g \mid s_i \in \mathcal{D}_g\}$ ($i=0,\ldots,N-1$), \emph{the upper bound of the tolerable poisoned samples} (denoted by $\overline{r}$) is 
\begin{equation} \vspace{-5pt}
    \overline{r}=\min |\Pi| \; s.t.\; |\bigcup_{i \in \Pi} \mathcal{S}_i| > G/2
\end{equation}
where $\Pi$ denotes a set of indices. The upper bound of the tolerable poisoned samples equals the minimum number of training samples that can influence more than a half of sub-classifiers.
\end{proposition}
The collective robustness must be zero when the poison budget $\geq \overline{r}$. We emphasize that computing $\overline{r}$ is an NP-hard max covering problem~\cite{fujishige2005submodular}. A simple way of enlarging $\overline{r}$ is to \emph{bound the influence scope for each sample} $|S_i|: i=0,\ldots, N-1$. In particular, if we bound the influence scope of each sample to be less than a constant $|S_i| \leq \Gamma: i=0,\ldots, N-1$ ($\Gamma$ is a constant), we have $\overline{r} \geq N/(2 \Gamma)$. This is the insight behind hash bagging.

\begin{figure}[t!]
  \centering
    \includegraphics[width=0.96\linewidth]{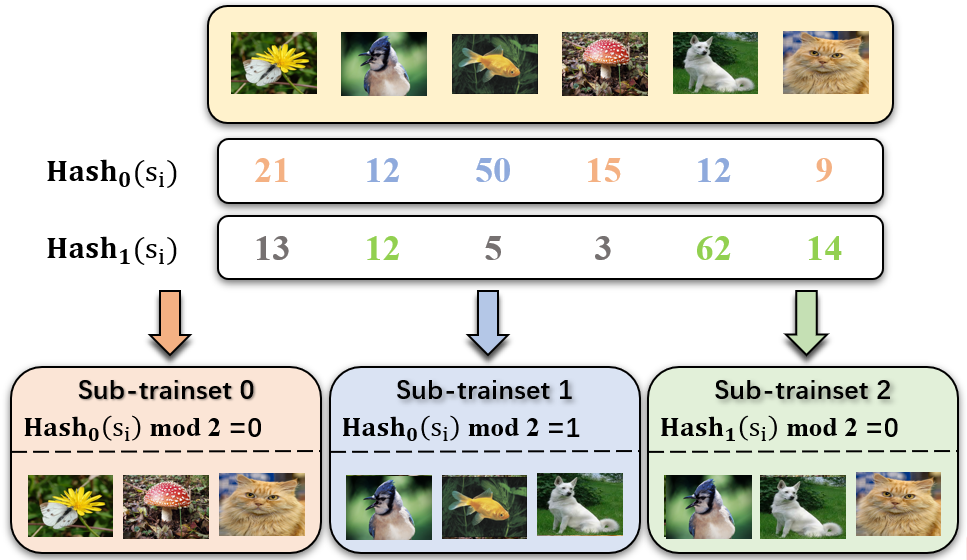}
    \caption{Hash bagging when $N=6$ (trainset size), $K=3$ (sub-trainset size), $G=3$ (number of sub-trainsets). $\hat{G}= \lfloor N/K \rfloor= 2$. By Eq. (\ref{eq:sub-trainset2}), the $0$-th sub-trainset ( $\hat{h}=0, \hat{g}=0$) is constructed based on $\rm{Hash}_{0}(s_i) \mod 2  = 0$ (the samples whose hash values are colored by \textcolor{lightred}{red}). The $1$-st sub-trainset ($\hat{h}=0, \hat{g}=1$) is constructed by $\rm{Hash}_{0}(s_i) \mod 2  = 1$ (the samples whose hash values are colored by \textcolor{lightblue}{blue}). The $2$-nd sub-trainset ($\hat{h}=1, \hat{g}=0$) is constructed by $\rm{Hash}_{1}(s_i) \mod 2  = 0$ (the samples whose hash values are colored by \textcolor{lightgreen}{green}).}
  \label{fig:subtrainset}
\end{figure}

\section{Proposed Approach: Hash Bagging}

\textbf{Objective of hash bagging.} We aim to improve vanilla bagging by designing a new subsampling algorithm. According to the remarks on Prop.~\ref{pro:bagging}, Prop.~\ref{pro:upper}, the new subsampling is expected to own the properties: \textbf{i) Determinism:} subsampling should be reproducible. \textbf{ii) Bounded influence scope:} inserting/deleting/modifying an arbitrary sample can only influence a limited number of sub-trainsets. \textbf{iii) Solvability:} the robustness can be computed within the given time. \textbf{iv) Generality:} the subsampling applies to arbitrary $K$ (the sub-trainset size) and $G$ (the number of sub-trainsets). 

The realization of hash bagging is based on the hash values. First let's see a simple case when $GK=N$.

\begin{algorithm2e}[tb!]
  \caption{Certify the collective robustness for our proposed hash bagging.}
  \label{alg:certify}
  \KwIn{testset $\mathcal{D}_{test}=\{x_j\}_{i=0}^{M-1}$, sub-classifiers $\{f_g\}_{g=1}^G$, the poison budget $r_{\rm{ins}},r_{\rm{del}},r_{\rm{mod}}$, sub-problem scale $\Delta$.}
  \For{$x_j: j=0,1,...,M-1$}{
  Compute predictions $\hat{y}_j= f_g(x_j):\, g=1,\ldots,G$;
  }
  \algcomment{See the simplification for $\mathbf{(P2)}$ (Eq.~\ref{eq:p2omega})}\\
  Compute the set of breakable predictions $\Omega$ \;
  \algcomment{Decompose the original problem to $\Delta$-scale sub-problems.}\\
  Decompose $\Omega= \bigcup_{\mu=0}^{\lceil M/\Delta \rceil-1} \mathcal{D}^{\mu}$, where $|\mathcal{D}^{\mu}|=\Delta$ ($\mu=0,\ldots, \lceil M/\Delta \rceil-2$) \; 
  \For{$\mathcal{D}^{\mu}: \mu=0,1,...\lceil M/\Delta \rceil-1$}{
  \algcomment{Solve the $\Delta$-scale sub-problems.}\\
  Compute the maximum number of simultaneously changed predictions $M^{\mu}_{\rm{ATK}}$ by solving $\mathbf{(P2)}$  over $\mathcal{D}^{\mu}$ w.r.t. the poison budget $r_{\rm{ins}},r_{\rm{del}},r_{\rm{mod}}$;
  }
  Compute the lower bound of the certified collective robustness: $M-\sum_{\mu} M_{\rm{ATK}}$\;
  \KwOut{$M-\sum_{\mu} M_{\rm{ATK}}$}
\end{algorithm2e}

\textbf{Hash bagging when $GK=N$.}
\label{sec:specific_hashsubsample}
Given $\mathcal{D}_{train}$, the $g$-th sub-trainset $\mathcal{D}_g$ ($g=0,1,\ldots, G-1$) is as follow:
\begin{equation}
    \label{eq:sub-trainset1}
\small
   \mathcal{D}_g= \{s_i \in \mathcal{D}_{train} \mid  \rm{Hash}(s_i) \mod G  =g \}
\end{equation}
 where $\rm{Hash}(\cdot)$ is the pre-specified hash function. Such that the number of sub-trainsets exactly equals $G$ and the sub-trainset size approximates $N/G=GK/G=K$, because the hash function will (approximately) uniformly allocate each sample to different hash values. Such hash-based subsampling satisfies the following properties: \textbf{i) Determinism}: fixing $G,K$, all $G$ sub-trainsets are uniquely determined by $\mathcal{D}_{train}$ and $\textrm{Hash}(\cdot)$, which we denoted as the trainset-hash pair $(\mathcal{D}_{train}, \textrm{Hash}(\cdot))$ for brevity. \textbf{ii) Bounded influence scope}: $r_{\rm{ins}}$ insertions, $r_{\rm{del}}$ deletions and $r_{\rm{mod}}$ modifications can influence at most $r_{\rm{ins}}+r_{\rm{del}}+2r_{\rm{mod}}$ sub-trainsets.

\textbf{Hash bagging for general cases.}\label{sec:hashbag}
Given $\mathcal{D}_{train}$ and a series of hash functions $\textrm{Hash}_{h}(\cdot)$  ($h=0, \ldots$), the $g$-th sub-trainset $\mathcal{D}_g$ ($g=0,1,\ldots, G-1$) is as follow:
\begin{equation}
\label{eq:sub-trainset2}
\small
    \mathcal{D}_g=\{s_i \in \mathcal{D}_{train} \mid \rm{Hash}_{\hat{h}}(s_i) \mod \hat{G}  =\hat{g} \}
\end{equation}
where $\hat{G}=\lfloor N/K \rfloor, \hat{h}= \lfloor g / \hat{G} \rfloor, \hat{g}= g \mod \hat{G}$. Specifically, we set $\hat{G}=\lfloor N/K \rfloor$, so that the size of each sub-trainset approximates $N/ \hat{G} \rightarrow K$. We specify a series of hash functions because that a trainset-hash pair can generate at most $\hat{G}$ sub-trainsets, thus we construct $\lceil G/ \hat{G} \rceil$ trainset-hash pairs, which is enough to generate $G$ sub-trainsets. Then the $g$-th sub-trainset is the $\hat{g}$-th sub-trainset within the sub-trainsets from the $\hat{h}$-th trainset-hash pair. Fig.~\ref{fig:subtrainset} illustratively shows an example of hash bagging. Remarkably, hash bagging satisfies: \textbf{i) Determinism}: the subsampling results only depends on the trainset-hash pairs $\{(\mathcal{D}_{train}, \textrm{Hash}_{h}(\cdot)): h=0, 1, \dots, \lceil G/ \hat{G} \rceil-1 \}$ if fixing $G,K$. \textbf{ii) Bounded influence scope}:  $r_{\rm{ins}}$ insertions, $r_{\rm{del}}$ deletions and $r_{\rm{mod}}$ modifications can influence at most $r_{\rm{ins}}+r_{\rm{del}}+2r_{\rm{mod}}$ sub-trainsets, within the $\hat{G}$ sub-trainsets from each trainset-hash pair. \textbf{iii) Generality}: hash bagging can be applied to all the combinations of $G, K$.

\textbf{Reproducible training of hash bagging.}
After constructing $G$ sub-trainsets based on Eq. (\ref{eq:sub-trainset2}), we train the sub-classifiers in a \emph{reproducible} manner. In our experiments, we have readily realized reproducibility by specifying the random seed for all the random operations.

\subsection{$\mathbf{(P2)}$: Collective Certification of Hash Bagging} 
\begin{proposition}[Simplified collective certification of hash bagging] \label{pro:hash}
For testset $\mathcal{D}_{test}=\{x_j\}_{j=0}^{M-1}$, we denote $\hat{y}_j= g(x_j)$ ($j=0,\ldots,M-1$) the ensemble prediction. The maximum number of simultaneously changed predictions (denoted by $M_{\rm{ATK}}$) under $r_{\rm{ins}}$ insertions, $r_{\rm{del}}$ deletions and $r_{\rm{mod}}$ modifications, is computed by $\mathbf{(P2)}$:
\begin{footnotesize}
\begin{align} 
    &\mathbf{(P2)}: \quad M_{\rm{ATK}}= \max_{A_0,\ldots,A_{G-1}} \; \sum_{x_j \in \mathcal{D}_{test}} \mathbb{I} \left\{ \ovo_{x_j}(\hat{y}_j) < \right. \nonumber\\ 
    & \left. \qquad \qquad \qquad \qquad \max_{y \neq \hat{y}_j} \left[\ovo_{x_j}(y)+ \frac{1}{2}\mathbb{I}\{y <\hat{y}_j\}\right] \right\} \label{eq:p2obj}\\
     & s.t. \quad [A_0, A_1, \ldots, A_{G-1}] \in \{0,1\}^G \label{eq:p2c0}\\
     & \sum_{g= (l-1) \hat{G}}^{\min\left(l \hat{G} -1, G\right)} A_g \leq r_{\rm{ins}}+r_{\rm{del}}+2r_{\rm{mod}} \nonumber \\
     & \qquad \qquad l=1,\ldots, \lceil G/\hat{G} \rceil \label{eq:p2c1}\\
     & \ovo_{x_j}(\hat{y}_j) = \underbrace{\vo_{x_j}(\hat{y}_j)}_{\textbf{Original votes}}- \underbrace{\sum_{g=1}^G A_g \mathbb{I}\{ f_g(x_j)=\hat{y}_j \}}_{\textbf{Influenced votes}}  \nonumber \\
     & \qquad \qquad \forall x_j \in \mathcal{D}_{test} \label{eq:p2c2}\\
     & \ovo_{x_j}(y)= \underbrace{\vo_{x_j}(y)}_{\textbf{Original votes}}+ \underbrace{\sum_{g=1}^G A_g \mathbb{I}\{ f_g(x_j) \neq y \}}_{\textbf{Influenced votes}} \nonumber \\
     & \qquad \qquad \forall x_j \in \mathcal{D}_{test},\; \forall y \neq \hat{y}_j \label{eq:p2c3}
\end{align}
\end{footnotesize}
The collective robustness is $M-M_{\rm{ATK}}$.
\end{proposition}
We now explain each equation respectively. Eq. (\ref{eq:p2obj}): the objective function is same as $\mathbf{(P1)}$. Eq. (\ref{eq:p2c0}): $A_1, A_2, \ldots, A_G$ are the binary variables represent the attack, where $A_g=1$ means that the $g$-th classifier is influenced. Eq. (\ref{eq:p2c1}): in hash bagging, $r_{\rm{ins}}$ insertions, $r_{\rm{del}}$ deletions and $r_{\rm{mod}}$ modifications can influence at most $r_{\rm{ins}}+r_{\rm{del}}+2r_{\rm{mod}}$ within each trainset-hash pair. Eq. (\ref{eq:p2c2}) and Eq. (\ref{eq:p2c3}): count the minimum/maximum number of votes (after being attacked) for $\hat{y}_j$ and $y \neq \hat{y}_j$. The main advantage of $\mathbf{(P2)}$ over $\mathbf{(P1)}$ is that, the size of the feasible region is reduced from $2^N$ to $2^G$ by exploiting the property of hash bagging, which significantly accelerates the solving process.

\subsection{Remarks on Proposition~\ref{pro:hash}}
\textbf{1) Tightness.} The collective robustness by $\mathbf{(P2)}$ is tight.

\textbf{2) Simplification.} $\mathbf{(P2)}$ can be simplified by ignoring the unbreakable predictions within the given poison budget. $\sum_{x_j \in \mathcal{D}_{test}}$ in Eq. (\ref{eq:p2obj}) can be simplified as $\sum_{x_j \in \Omega}$, and $\Omega$:
    \begin{footnotesize}
    \begin{align} \label{eq:p2omega} 
       \Omega=&\{x_j \in  \mathcal{D}_{test}: \vo_{x_j}(\hat{y}_j)-\max_{y \neq \hat{y}_j} \left[ \vo_{x_j}(y)+\mathbb{I}\{y <\hat{y}_j\} \right] \nonumber\\
        \leq &2  \lceil G/ \hat{G} \rceil(r_{\rm{ins}}+r_{\rm{del}}+2r_{\rm{mod}})\} 
    \end{align}
    \end{footnotesize}

\textbf{3) NP-hardness.} $\mathbf{(P2)}$ is NP-hard. We can speedup the solution process by decomposition (see Section~\ref{sec:decomposition}).
 

\textbf{Implementation.} Alg.~\ref{alg:certify} shows our algorithm for certifying collective robustness. Specifically, we apply simplification and decomposition to accelerate solving $\mathbf{(P2)}$.

\textbf{Compare hash bagging to vanilla bagging.} In Fig.~\ref{fig:ensemble} and Fig.~\ref{fig:r_K}, we compare hash bagging to vanilla bagging on the ensemble accuracy and $\overline{r}$ (see Prop.~\ref{pro:upper}) respectively. We observe in Fig.~\ref{fig:ensemble} that the ensemble accuracy of hash bagging roughly equals vanilla bagging. Notably, the accuracy variance of hash bagging (over different hash functions) is much smaller than vanilla bagging. We observe in Fig.~\ref{fig:r_K} that $\overline{r}$ of hash bagging is consistently higher than vanilla bagging, especially when $K$ is small. The comparisons suggest that, hash bagging is much more robust than vanilla bagging without sacrificing the ensemble accuracy.

\begin{figure}[tb!]
\centering
  \begin{subfigure}[b]{0.45\textwidth}
  \centering
    \includegraphics[width=.8\linewidth]{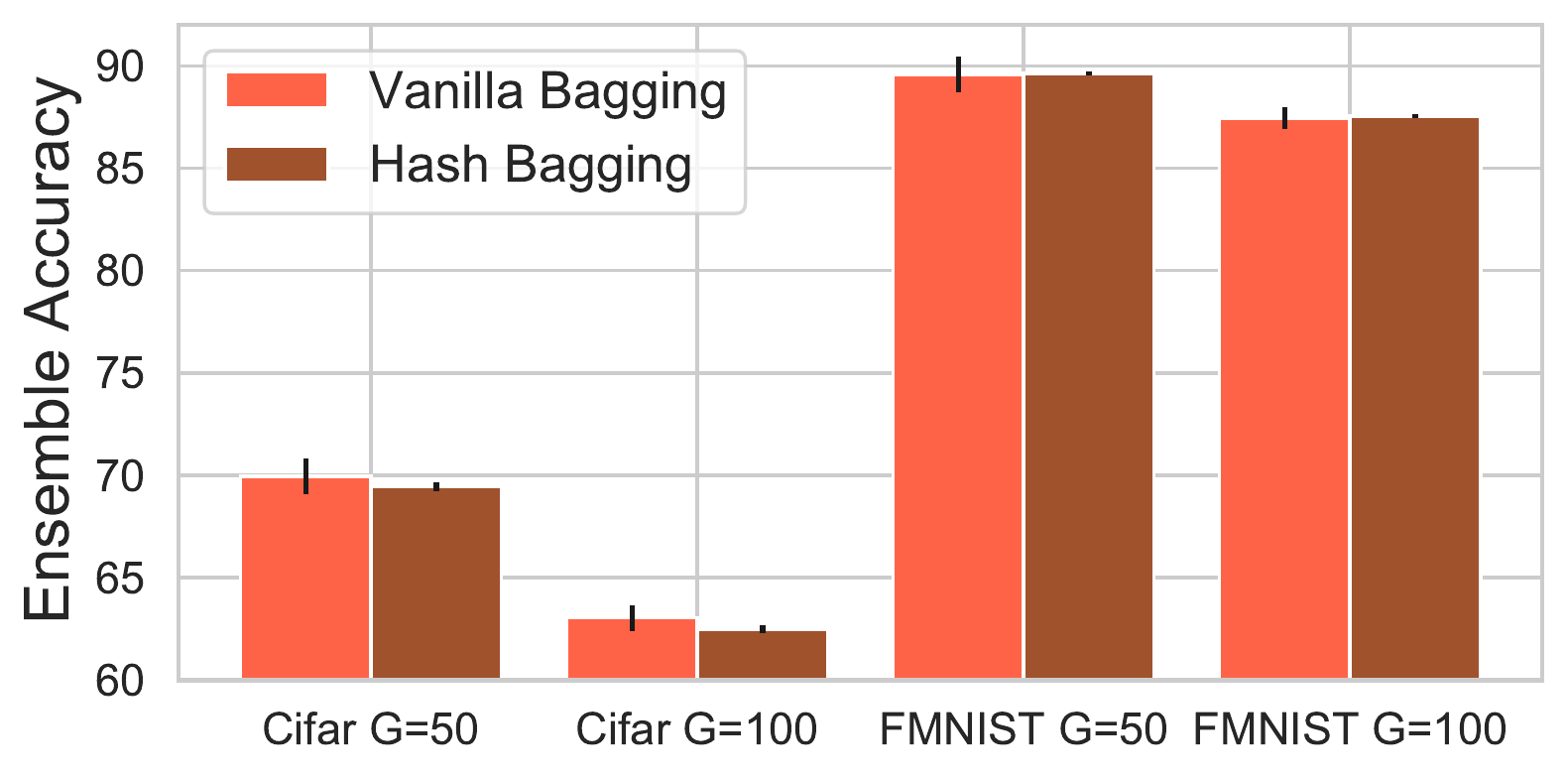}
    \caption{Comparison on ensemble accuracy ($K= N/G$).}
    \label{fig:ensemble}
      \end{subfigure}
      \\
  \begin{subfigure}[b]{0.45 \textwidth}
  \centering
    \includegraphics[width=.8\linewidth]{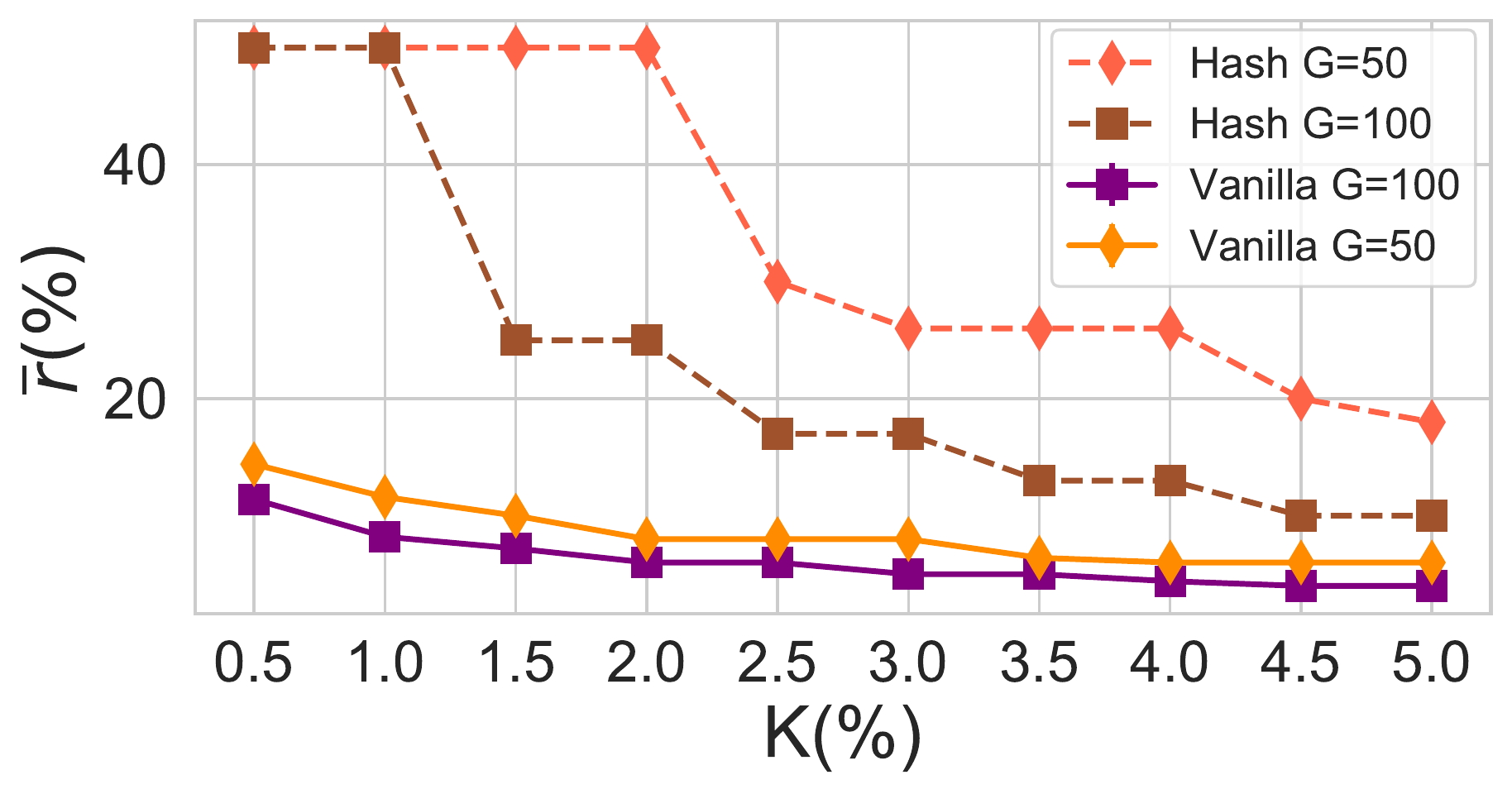}
    \caption{Comparison on $\overline{r}$ on FMNIST.}
    \label{fig:r_K}
  \end{subfigure}
\caption{Comparing hash bagging to vanilla bagging.}
\end{figure}

\section{Comparisons to Prior Works}
We compare to prior works that are tailored to the general data poisoning attack~\cite{ma2019data,levine2021deep,jia2021intrinsic,jia2022rnn}.

\textbf{Comparison to ~\cite{ma2019data}}
Compared to differential privacy based defense~\cite{ma2019data}, hash bagging is more practical for two reasons: I) hash bagging does not require the training algorithm to be differentially private. II) The differential privacy often harms the performance of the learnt model~\cite{duchi2013local}, which also limits the scalability of this type of defenses.

\textbf{Comparison to ~\cite{jia2022rnn}}
Compared to ~\cite{jia2022rnn} which derives the sample-wise/collective certificates for kNN/rNN, hash bagging is compatible with different model architectures. Note that the effectiveness of kNN/rNN relies on the assumption: close data are typically similar. Since this assumption might do not hold in some classification tasks, we believe hash bagging is much more practical.

\textbf{Comparison to ~\cite{jia2021intrinsic}}
~\cite{jia2021intrinsic} proposes a bagging variant as a certified defense, which predicts the majority class among the predictions of all the possible sub-classifiers (total $N^K$ sub-classifiers). In practice, training $N^K$ sub-classifiers is often unaffordable, ~\cite{jia2021intrinsic} approximately estimates the voting distribution by a confidence interval method, which needs to train hundreds of sub-classifiers for a close estimate ($G$ is required to be large). In comparison, hash bagging has no additional constraint. Moreover, unlike our deterministic robustness certificates, its robustness certificates are probabilistic, which have an inevitable failure probability.

\textbf{Comparison to~\cite{levine2021deep}}
~\cite{levine2021deep} propose a partition-based bagging as a certified defense, which is corresponding to \textbf{Hash subsampling when $GK=N$} (Section~\ref{sec:specific_hashsubsample}). In comparison, both our collective certification and hash bagging are more general than ~\cite{levine2021deep}. Specifically, hash bagging ablates the constraint that ~\cite{levine2021deep} places on the bagging hyper-parameters $G,K$. Our collective certification is able to certify both the tight collective robustness and sample-wise robustness, while ~\cite{levine2021deep} only considers the sample-wise certificate.

\begin{table}[tb!]
\centering
\caption{Experimental setups in line with literature.}
\label{tab:config}
    \begin{adjustbox}{width=0.78\linewidth}
    \begin{tabular}{ c | c | c | c | c }
    \toprule
     Dataset & Trainset & Testset & Class & Classifier \\
     \hline
     Bank & 35,211 & 10,000  & 2 & Bayes\\
     \hline
     Electricity & 35,312 & 10,000  & 2 & SVM \\
     \hline
     FMNIST & 60,000 & 10,000 & 10  & NIN\\
     \hline
     CIFAR-10 & 50,000 & 10,000 & 10 & NIN (Augmentation)\\
    \bottomrule
    \end{tabular}   
    \end{adjustbox}
\end{table}

\section{Experiments}
\subsection{Experimental Setups} 
\textbf{Datasets and models.} We evaluate hash bagging and collective certification on two classic machine learning datasets: Bank~\cite{moro2014data}, Electricity~\cite{harries1999splice}, and two image classification datasets: FMNIST~\cite{xiao2017fashion}, CIFAR-10~\cite{krizhevsky2009learning}. Specifically, for Bank and Electricity, we adapt vanilla bagging/hash bagging to the machine-learning models: Bayes and SVM. For FMNIST and CIFAR-10, we adapt vanilla bagging/hash bagging to the deep-learning model Network in Network (NiN) ~\cite{min2014nin}. The detailed experimental setups are shown in Table ~\ref{tab:config}. 

\textbf{Implementation details.} We use Gurobi 9.0~\cite{gurobi} to solve $\mathbf{(P1)}$ and $\mathbf{(P2)}$, which can return a lower/upper bound of the objective value within the pre-specific time period. Generally, a longer time can yield a tighter bound. For efficiency, we limit the time to be $2$s per sample\footnote{The solving time for $\mathbf{(P1)}$ is universally set to be $2 |\mathcal{D}_{test}|= 20,000$ seconds. The solving time for $\mathbf{(P2)}$ is set to be $2|\Omega|$ for $\mathbf{(P2)}$ where $\Omega$ is defined in Eq. (\ref{eq:p2omega}).}. More details are in Appendix (Section~\ref{sec:implement}).

\textbf{Evaluation metrics and peer methods.} Following~\cite{levine2021deep,jia2021intrinsic,jia2022rnn}, we evaluate the performance by two metrics: \emph{collective robustness} and \emph{certified accuracy}\footnote{We report the minimum number of accurate predictions as the certified accuracy, instead of a ratio, which is in line with the practice in the literature of collective robustness.}. We also report the relative gap (denoted by \down{\alpha}) between the maximum number of simultaneously changed (correct) predictions guaranteed by the collective certification (denoted by $M^{\rm{col}}_{\rm{ATK}}$) and that of the sample-wise certification (denoted by $M^{\rm{sam}}_{\rm{ATK}}$). Namely, \down{\alpha}$=( M^{\rm{sam}}_{\rm{ATK}}-M^{\rm{col}}_{\rm{ATK}}) / M^{\rm{sam}}_{\rm{ATK}}$. High $\alpha$ means that the sample-wise certification highly over-estimates the poisoning attack. All the experiments are conducted on the clean dataset without being attacked, which is a common experimental setting for certified defenses~\cite{levine2021deep,jia2021intrinsic,jia2022rnn}. We compare \emph{hash bagging} to \emph{vanilla bagging}, and compare \emph{collective certification} to sample-wise certification~\cite{levine2021deep}. We also compare to \emph{probabilistic certification} ~\cite{jia2021intrinsic} in Appendix (Section~\ref{sec:prob}).

\begin{table}[tb!]
\centering
\caption{(Bank: $M=10,000$; $K= 5\% N$) Certified collective robustness and certified accuracy at $r=5\%, \ldots, 25\%$ ($\times G$). $r$ refers to the poison budget $r= r_{\rm{ins}}+r_{\rm{del}}+2 r_{\rm{mod}}$. \textbf{Sample-wise}: sample-wise certification. \textbf{Collective}: collective certification. \textbf{CR} and \textbf{CA}: certified collective robustness and certified accuracy. \down{\alpha}: the relative gap between $M_{\rm{ATK}}$ guaranteed by collective certification and $M_{\rm{ATK}}$ of  sample-wise certification. NaN: division by zero. 
} \label{tab:bank}
\begin{adjustbox}{width=0.98\linewidth}
    \begin{tabular}{lll | llcccccccccc}
    \toprule
   G & Bagging & Certification & Metric & $5\%$ & $10\%$ & $15\%$ & $20\%$ & $25\%$\\
    \midrule
   \multirow{15}{*}{20} 
   & \multirow{6}{*}{Vanilla} &  \multirow{2}{*}{Sample-wise} & CR & 3917 & 0 & 0 & 0 & 0 \\
    \cmidrule(l){4-9}
    & & &  CA & 3230 & 0 & 0 & 0 & 0 \\
    \cmidrule(l){3-9}
    &  & \multirow{4}{*}{Collective} & CR & 4449 & 0 & 0 & 0 & 0 \\
    & & & $M_{\rm{ATK}}$ & \down{8.74} & NaN & NaN & NaN & NaN\\
    \cmidrule(l){4-9}
    & & &  CA & 3588 & 0 & 0 & 0 & 0 \\
    & & & $M_{\rm{ATK}}$ & \down{7.47} & NaN & NaN & NaN & NaN\\
    \cmidrule(l){2-9}
    & \multirow{6}{*}{Hash} &  \multirow{2}{*}{Sample-wise} & CR & 9599 & 9009 & 7076 & 5778 & 4686 \\
    \cmidrule(l){4-9}
    & & &  CA & 7788 & 7403 & 5755 & 4644 & 3817 \\
    \cmidrule(l){3-9}
    &  & \multirow{4}{*}{Collective} & CR &  \textbf{9718} & \textbf{9209} & \textbf{7270} & \textbf{5968} & \textbf{4930} \\
    & & & $M_{\rm{ATK}}$ & \down{29.7} & \down{20.2} & \down{6.63} & \down{4.50} & \down{4.59}\\
    \cmidrule(l){4-9}
    & & &  CA  & \textbf{7831} & \textbf{7464} & \textbf{5806} & \textbf{4685} & \textbf{3881}  \\
    & & & $M_{\rm{ATK}}$ & \down{18.5} & \down{9.89} & \down{2.25} & \down{1.21} & \down{1.52}\\
    \hline
    \multirow{15}{*}{40} 
    & \multirow{6}{*}{Vanilla} &  \multirow{2}{*}{Sample-wise} & CR & 5250 & 1870 & 0 & 0 & 0 \\
    \cmidrule(l){4-9}
    & & &  CA & 4160 & 1408 & 0 & 0 & 0 \\
    \cmidrule(l){3-9}
    &  & \multirow{4}{*}{Collective} & CR & 5385 & 2166 & 0 & 0 & 0 \\
    & & & $M_{\rm{ATK}}$ & \down{2.84} & \down{3.64} & NaN & NaN & NaN\\
    \cmidrule(l){4-9}
    & & &  CA & 4190 & 1647 & 0 & 0 & 0 \\
    & & & $M_{\rm{ATK}}$ & \down{0.77} & \down{3.58} & NaN & NaN & NaN\\
    \cmidrule(l){2-9}
    & \multirow{6}{*}{Hash} &  \multirow{2}{*}{Sample-wise} & CR & 9638 & 9301 & 6401 & 5376 & 4626 \\
    \cmidrule(l){4-9}
    & & &  CA & 7881 & 7679	& 5198 & 4354 & 3718 \\
    \cmidrule(l){3-9}
    &  & \multirow{4}{*}{Collective} & CR &  \textbf{9762} & \textbf{9475} & \textbf{6603} & \textbf{5572} & \textbf{4796} \\
    & & & $M_{\rm{ATK}}$ & \down{34.2} & \down{24.9} & \down{5.61} & \down{4.24} & \down{3.16}\\
    \cmidrule(l){4-9}
    & & &  CA & \textbf{7914} & \textbf{7718} & \textbf{5236} & \textbf{4396}	& \textbf{3751} \\
    & & & $M_{\rm{ATK}}$ & \down{17.2} & \down{9.90} & \down{1.32} & \down{1.13} & \down{0.76}\\
    \bottomrule
    \end{tabular}
    \end{adjustbox}
\end{table}

\subsection{Experimental Results}
\textbf{Bank and Electricity.} Table~\ref{tab:electricity} and Table~\ref{tab:bank} report the performances of sample-wise/collective certification on vanilla/hash bagging. There is no need to apply decomposition to these two binary-classification datasets since we can compute the tight certified collective robustness within $10^2$ seconds. In comparison, the collective robustness of vanilla bagging drops to zero at $r =15\% G$, while hash bagging is able to achieve a non-trivial collective robustness at $r = 25 \% G$. The values of \down{\alpha} demonstrate that the exact value of $M_{\rm{ATK}}$ is $5\% \sim 30\%$ less than the values derived from the sample-wise certification. There is an interesting phenomenon that \down{\alpha} generally decreases with $r$ for the number of the candidate poisoning attacks ${N \choose r}$ exponentially increases with $r$. When $r$ is large, there is a high probability to find an attack that can corrupt a high percent of the breakable predictions, thus $M_{\rm{ATK}}$ guaranteed by the collective certification is close to the sample-wise certification. As we can see, the collective robustness/certified accuracy at $G=20$ are roughly equal to that of $G=40$. This is because an insertion/deletion is considered to influence 1 ($5\%$) vote among total 20 votes when $G=20$, while it can influence 2 ($5\%$) votes among 40 votes for the sub-trainset overlapping. Since the voting distribution of $G=20$ and $G=40$ are similar, $G=20$ and $G=40$ own the similar collective robustness. 

\begin{table}[tb!]
\centering
\caption{(Electricity: $M=10,000$; $K= 5\% N$) Certified collective robustness and certified accuracy.} \label{tab:electricity}
\begin{adjustbox}{width=0.98\linewidth}
    \begin{tabular}{lll | llcccccccccc}
    \toprule
   G & Bagging & Certification & Metric & $5\%$ & $10\%$ & $15\%$ & $20\%$ & $25\%$\\
    \midrule
   \multirow{15}{*}{20} 
   & \multirow{6}{*}{Vanilla} &  \multirow{2}{*}{Sample-wise} & CR & 9230 & 0 & 0 & 0 & 0 \\
    \cmidrule(l){4-9}
    & & &  CA & 7321 & 0 & 0 & 0 & 0 \\
    \cmidrule(l){3-9}
    &  & \multirow{4}{*}{Collective} & CR & 9348 & 0 & 0 & 0 & 0 \\
    & & & $M_{\rm{ATK}}$ & \down{15.3} & NaN & NaN & NaN & NaN\\
    \cmidrule(l){4-9}
    & & &  CA & 7394 & 0 & 0 & 0 & 0 \\
    & & & $M_{\rm{ATK}}$ & \down{17.5} & NaN & NaN & NaN & NaN\\
    \cmidrule(l){2-9}
    & \multirow{6}{*}{Hash} &  \multirow{2}{*}{Sample-wise} & CR & 9858 & 9738 & 9602 & 9461 & 9293 \\
    \cmidrule(l){4-9}
    & & &  CA & 7681 & 7621	& 7538	& 7462	& 7362 \\
    \cmidrule(l){3-9}
    &  & \multirow{4}{*}{Collective} & CR &  \textbf{9915} & \textbf{9821} & \textbf{9726} & \textbf{9608} & \textbf{9402} \\
    & & & $M_{\rm{ATK}}$ & \down{40.1} & \down{31.7} & \down{31.1} & \down{27.3} & \down{23.9}\\
    \cmidrule(l){4-9}
    & & &  CA  & \textbf{7701} & \textbf{7663}	& \textbf{7608} & \textbf{7547} & \textbf{7458}  \\
    & & & $M_{\rm{ATK}}$ & \down{34.5} & \down{35.6} & \down{34.8} & \down{30.7} & \down{25.5}\\
    \hline
    \multirow{15}{*}{40} 
    & \multirow{6}{*}{Vanilla} &  \multirow{2}{*}{Sample-wise} & CR & 9482 & 8648 & 0 & 0 & 0 \\
    \cmidrule(l){4-9}
    & & &  CA & 7466 & 6986 & 0 & 0 & 0 \\
    \cmidrule(l){3-9}
    &  & \multirow{4}{*}{Collective} & CR & 9566 & 8817 & 0 & 0 & 0 \\
    & & & $M_{\rm{ATK}}$ & \down{16.2} & \down{12.5} & NaN & NaN & NaN\\
    \cmidrule(l){4-9}
    & & &  CA & 7513 & 7086 & 0 & 0 & 0 \\
    & & & $M_{\rm{ATK}}$ & \down{16.5} & \down{13.1} & NaN & NaN & NaN\\
    \cmidrule(l){2-9}
    & \multirow{6}{*}{Hash} &  \multirow{2}{*}{Sample-wise} & CR & 9873 & 9769 & 9636 & 9491 & 9366 \\
    \cmidrule(l){4-9}
    & & &  CA & 7681 & 7625	& 7546 & 7459 & 7399 \\
    \cmidrule(l){3-9}
    &  & \multirow{4}{*}{Collective} & CR &  \textbf{9919} & \textbf{9842} & \textbf{9755} & \textbf{9601} & \textbf{9461} \\
    & & & $M_{\rm{ATK}}$ &  \down{36.2} & \down{31.6} & \down{32.7} & \down{21.6} & \down{15.0}\\
    \cmidrule(l){4-9}
    & & &  CA & \textbf{7700} & \textbf{7661}	& \textbf{7613}	& \textbf{7536}	& \textbf{7457} \\
    & & & $M_{\rm{ATK}}$ &  \down{27.5} & \down{28.8} & \down{32.8} & \down{26.5} & \down{16.5}\\
    \bottomrule
    \end{tabular}
    \end{adjustbox}
\end{table}

\textbf{FMNIST and CIFAR-10.} Table~\ref{tab:fmnist} and Table~\ref{tab:cifar} report the performance of sample-wise/collective certification (with/without decomposition) on vanilla/hash bagging. We adapt decomposition for speedup, because $\mathbf{(P1)}$ and $\mathbf{(P2)}$ are not solvable over those two ten-classes classification datasets within the limited time. The $\Delta$ choices are reported in Appendix (Section~\ref{sec:delta}). We see that hash bagging consistently outperforms vanilla bagging across different poison budgets. The results demonstrate that: collective certification with decomposition $>$ collective certification $>$ sample-wise certification in terms of the certified collective robustness and the certified accuracy, which suggests collective certification with decomposition is an efficient way to compute the collective robustness certificate.

\begin{table}[tb!]
\centering
\caption{(FMNIST: $M=10,000$; $K= N/G$) Certified collective robustness and certified accuracy. \textbf{Decomposition}: collective certification with decomposition.} \label{tab:fmnist}
\begin{adjustbox}{width=0.98\linewidth}
    \begin{tabular}{lll | llcccccccccc}
    \toprule
   G & Bagging & Certification & Metric & $5\%$ & $10\%$ & $15\%$ & $20\%$ & $25\%$\\
    \midrule
   \multirow{15}{*}{50} 
   & \multirow{6}{*}{Vanilla} &  \multirow{2}{*}{Sample-wise} & CR & 7432 & 0 & 0 & 0 & 0 \\
    \cmidrule(l){4-9}
    & & &  CA & 7283 & 0 & 0 & 0 & 0 \\
    \cmidrule(l){3-9}
    &  & \multirow{4}{*}{Collective} & CR & 7727 & 0 & 0 & 0 & 0 \\
    & & & $M_{\rm{ATK}}$ & \down{11.5} & NaN & NaN & NaN & NaN\\
    \cmidrule(l){4-9}
    & & &  CA & 7515 & 0 & 0 & 0 & 0 \\
    & & & $M_{\rm{ATK}}$ & \down{13.8} & NaN & NaN & NaN & NaN\\
    \cmidrule(l){2-9}
    & \multirow{10}{*}{Hash} &  \multirow{2}{*}{Sample-wise} & CR & 9576 & 9307	& 8932 & 8671 & 8238 \\
    \cmidrule(l){4-9}
    & & &  CA & 8768 & 8635	& 8408 & 8246 & 7943 \\
    \cmidrule(l){3-9}
    &  & \multirow{4}{*}{Collective} & CR & \textbf{9726}	& 9410 & 9024 & 8761 & 8329 \\
    & & & $M_{\rm{ATK}}$ & \down{35.4} & \down{14.9} & \down{8.61} & \down{6.77} & \down{5.16}\\
    \cmidrule(l){4-9}
    & & &  CA & \textbf{8833} & \textbf{8719} & 8493 & 8327 & 8022 \\
    & & & $M_{\rm{ATK}}$ & \down{32.8} & \down{25.4} & \down{15.2} & \down{11.2} & \down{7.72}\\
    \cmidrule(l){3-9}
    &  & \multirow{4}{*}{Decomposition} & CR & 9666 & \textbf{9472} & \textbf{9124} & \textbf{8887} & \textbf{8491} \\
    & & & $M_{\rm{ATK}}$ & \down{21.2} & \down{23.8} & \down{18.0} & \down{16.2} & \down{14.4}\\
    \cmidrule(l){4-9}
    & & &  CA & 8812 & 8716 & \textbf{8527} & \textbf{8385} & \textbf{8119} \\
    & & & $M_{\rm{ATK}}$ & \down{22.2} & \down{24.5} & \down{21.3} & \down{19.3} & \down{17.2}\\
    \hline
    \multirow{15}{*}{100} 
    & \multirow{6}{*}{Vanilla} &  \multirow{2}{*}{Sample-wise} & CR & 7548 & 0 & 0 & 0 & 0 \\
    \cmidrule(l){4-9}
    & & &  CA & 7321 & 0 & 0 & 0 & 0 \\
    \cmidrule(l){3-9}
    &  & \multirow{4}{*}{Collective} & CR & 8053 & 0 & 0 & 0 & 0 \\
    & & & $M_{\rm{ATK}}$ & \down{20.6} & NaN & NaN & NaN & NaN\\
    \cmidrule(l){4-9}
    & & &  CA & 7746 & 0 & 0 & 0 & 0 \\
    & & & $M_{\rm{ATK}}$ & \down{29.4} & NaN & NaN & NaN & NaN\\
    \cmidrule(l){2-9}
    & \multirow{10}{*}{Hash} &  \multirow{2}{*}{Sample-wise} & CR & 9538 & 9080 & 8653 & 8249 & 7823 \\
    \cmidrule(l){4-9}
    & & &  CA & 8554 & 8316 & 8049 & 7797	& 7486 \\
    \cmidrule(l){3-9}
    &  & \multirow{4}{*}{Collective} & CR & 9611	& 9167 & 8754	& 8344 & 7912  \\
    & & & $M_{\rm{ATK}}$ & \down{15.8} & \down{9.46} & \down{7.50} & \down{5.42} & \down{4.09}\\
    \cmidrule(l){4-9}
    & & &  CA & \textbf{8610}	& 8375 & 8116 & 7857 & 7558 \\
    & & & $M_{\rm{ATK}}$ & \down{26.7} & \down{13.2} & \down{9.37} & \down{6.20} & \down{5.63}\\
    \cmidrule(l){3-9}
    &  & \multirow{4}{*}{Decomposition} & CR & \textbf{9631} & \textbf{9232} & \textbf{8837} & \textbf{8450} & \textbf{8036} \\
    & & & $M_{\rm{ATK}}$ & \down{20.1} & \down{16.5} & \down{13.6} & \down{11.5} & \down{9.78}\\
    \cmidrule(l){4-9}
    & & &  CA & 8595 & \textbf{8407} & \textbf{8152} & \textbf{7917} & \textbf{7639} \\
    & & & $M_{\rm{ATK}}$ & \down{19.5} & \down{20.3} & \down{14.4} & \down{12.4} & \down{12.0}\\
    \bottomrule
    \end{tabular}
    \end{adjustbox}
\end{table}

\subsection{Ablation Study}
\textbf{Impact of $G$.}
Fig.~\ref{fig:ablate_g} reports the impact of $G$ on the certified collective robustness of hash bagging. The figure illustrates that as $G$ increases, the collective robustness increases first and then decreases, which reaches the top at $GK=N$. The reason is, as $G$ increases to $N/K$, the total number of votes increases, thus the attacker needs to modify more votes (higher poison budget) to modify the majority class. As $G$ exceeds the threshold of $N/K$, despite the growing number of votes, the influence scope of a poisoned sample also increases, as an insertion can simultaneously influence two sub-trainsets when $KG>N$, which causes a slight decline on the certified collective robustness.

\begin{table}[tb!]
\centering
\caption{(CIFAR-10: $M=10,000$; $K= N/G$) Certified collective robustness and certified accuracy.} \label{tab:cifar}
\begin{adjustbox}{width=0.98\linewidth}
    \begin{tabular}{lll | llcccccccccc}
    \toprule
   G & Bagging & Certification & Metric & $5\%$ & $10\%$ & $15\%$ & $20\%$ & $25\%$\\
    \midrule
   \multirow{15}{*}{50} 
   & \multirow{6}{*}{Vanilla} &  \multirow{2}{*}{Sample-wise} & CR & 2737 & 0 & 0 & 0 & 0 \\
    \cmidrule(l){4-9}
    & & &  CA & 2621 & 0 & 0 & 0 & 0 \\
    \cmidrule(l){3-9}
    &  & \multirow{4}{*}{Collective} & CR & 3621 & 0 & 0 & 0 & 0 \\
    & & & $M_{\rm{ATK}}$ & \down{12.2} & NaN & NaN & NaN & NaN\\
    \cmidrule(l){4-9}
    & & &  CA & 3335 & 0 & 0 & 0 & 0 \\
    & & & $M_{\rm{ATK}}$ & \down{16.3} & NaN & NaN & NaN & NaN \\
    \cmidrule(l){2-9}
    & \multirow{10}{*}{Hash} &  \multirow{2}{*}{Sample-wise} & CR & 8221 & 7268 & 6067 & 5320 & 4229 \\
    \cmidrule(l){4-9}
    & & &  CA & 6305 & 5864 & 5186 & 4705 & 3884 \\
    \cmidrule(l){3-9}
    &  & \multirow{4}{*}{Collective} & CR & 8393	& 7428 & 6204	& 5435 & 4290  \\
    & & & $M_{\rm{ATK}}$ & \down{9.67} & \down{5.86} & \down{3.48} & \down{2.46} & \down{1.06}\\
    \cmidrule(l){4-9}
    & & &  CA & 6410	& 5985 & 5342	& 4848 & 4006 \\
    & & & $M_{\rm{ATK}}$ & \down{15.2} & \down{10.7} & \down{8.62} & \down{6.24} & \down{3.92}\\
    \cmidrule(l){3-9}
    &  & \multirow{4}{*}{Decomposition} & CR & \textbf{8694} & \textbf{7854} & \textbf{6686} & \textbf{5912} & \textbf{4826} \\
    & & & $M_{\rm{ATK}}$ & \down{26.6} & \down{21.4} & \down{15.7} & \down{12.6} & \down{10.3}\\
    \cmidrule(l){4-9}
    & & &  CA & \textbf{6490} & \textbf{6147} & \textbf{5553} & \textbf{5113} & \textbf{4341} \\
    & & & $M_{\rm{ATK}}$ & \down{26.8} & \down{25.0} & \down{20.2} & \down{17.8} & \down{14.7}\\
    \hline
    \multirow{15}{*}{100} 
    & \multirow{6}{*}{Vanilla} &  \multirow{2}{*}{Sample-wise} & CR & 2621 & 0 & 0 & 0 & 0 \\
    \cmidrule(l){4-9}
    & & &  CA & 1876 & 0 & 0 & 0 & 0 \\
    \cmidrule(l){3-9}
    &  & \multirow{4}{*}{Collective} & CR & 2657 & 0 & 0 & 0 & 0 \\
    & & & $M_{\rm{ATK}}$ & \down{7.93} & NaN & NaN & NaN & NaN\\
    \cmidrule(l){4-9}
    & & &  CA & 2394 & 0 & 0 & 0 & 0 \\
    & & & $M_{\rm{ATK}}$ & \down{11.8} & NaN & NaN & NaN & NaN\\
    \cmidrule(l){2-9}
    & \multirow{10}{*}{Hash} &  \multirow{2}{*}{Sample-wise} & CR & 7685 & 5962 & 4612 & 3504 & 2593 \\
    \cmidrule(l){4-9}
    & & &  CA & 5396	& 4571 & 3787 & 3008 & 2315 \\
    \cmidrule(l){3-9}
    &  & \multirow{4}{*}{Collective} & CR & 7744	& 5974 & 4618	& 3509 & 2598 \\
    & & & $M_{\rm{ATK}}$ & \down{2.54} & \down{0.30} & \down{0.11} & \down{0.08} & \down{0.07} \\
    \cmidrule(l){4-9}
    & & &  CA & 5475	& 4650	& 3825	& 3030	& 2330 \\
    & & & $M_{\rm{ATK}}$ & \down{9.21} & \down{4.69} & \down{1.54} & \down{0.68} & \down{0.38}\\
    \cmidrule(l){3-9}
    &  & \multirow{4}{*}{Decomposition} & CR & \textbf{8137} & \textbf{6469} & \textbf{5061} & \textbf{4035} & \textbf{2987} \\
    & & & $M_{\rm{ATK}}$ & \down{19.5} & \down{12.5} & \down{8.33} & \down{8.17} & \down{5.32}\\
    \cmidrule(l){4-9}
    & & &  CA & \textbf{5570} & \textbf{4841} & \textbf{4098} & \textbf{3338} & \textbf{2635} \\
    & & & $M_{\rm{ATK}}$ & \down{20.3} & \down{16.0} & \down{12.6} & \down{10.2} & \down{8.12}\\
    \bottomrule
    \end{tabular}
    \end{adjustbox}
\end{table}

\textbf{Impact of $K$.}
Fig.~\ref{fig:ablate_k} reports the impact of $K$ on the certified collective robustness of hash bagging. Similar to $G$, as $K$ increases, the collective robustness increases first till $K=N/G$ and then decreases. The insight is, as $K$ rises to $N/G$, the collective robustness first increases for the improved prediction accuracy of each sub-classifier, because all the sub-classifiers have a higher probability to predict the correct class, as validated in Fig.~\ref{fig:votes_k}. As $K$ exceeds the threshold of $N/G$, the collective robustness decreases for the overlapping between the sub-trainsets, with the same reason of $G$. 

\textbf{Impact of sub-testset scale $\Delta$.}
 Fig.~\ref{fig:ablate_d1} and Fig.~\ref{fig:ablate_d2} report the impact of $\Delta$ on the certified collective robustness of hash bagging at $r=15\% G$. Specifically, Fig.~\ref{fig:ablate_d1} reports the impact of $\Delta$ at no time limit, where we can compute the tight collective robustness for each $\Delta$-size sub-testset. As shown in the figure, the certified collective robustness grows with $\Delta$, but higher $\Delta$ also enlarges the computation cost. Thus, $\Delta$ controls the trade-off between the collective robustness and the computation cost. Fig.~\ref{fig:ablate_d1} shows the impact of $\Delta$ when the time is limited by 2s per sample. We observe that the robustness first increases with $\Delta$ and then decreases. The increase is for that we can compute the optimal objective value when $\Delta$ is low, and the computed collective robustness lower bound increases with $\Delta$ as validated in Fig.~\ref{fig:ablate_d1}. The decrease is because that the required time for solving $\mathbf{(P2)}$ is exponential to $\Delta$. Consequently, we can only obtain a loose bound that is far from the optimal value within the limited time, which causes the decline on the certified collective robustness.

\begin{figure}[tb!]
\centering
  \begin{subfigure}[b]{0.25 \textwidth}
  \centering
    \includegraphics[width=1.\linewidth]{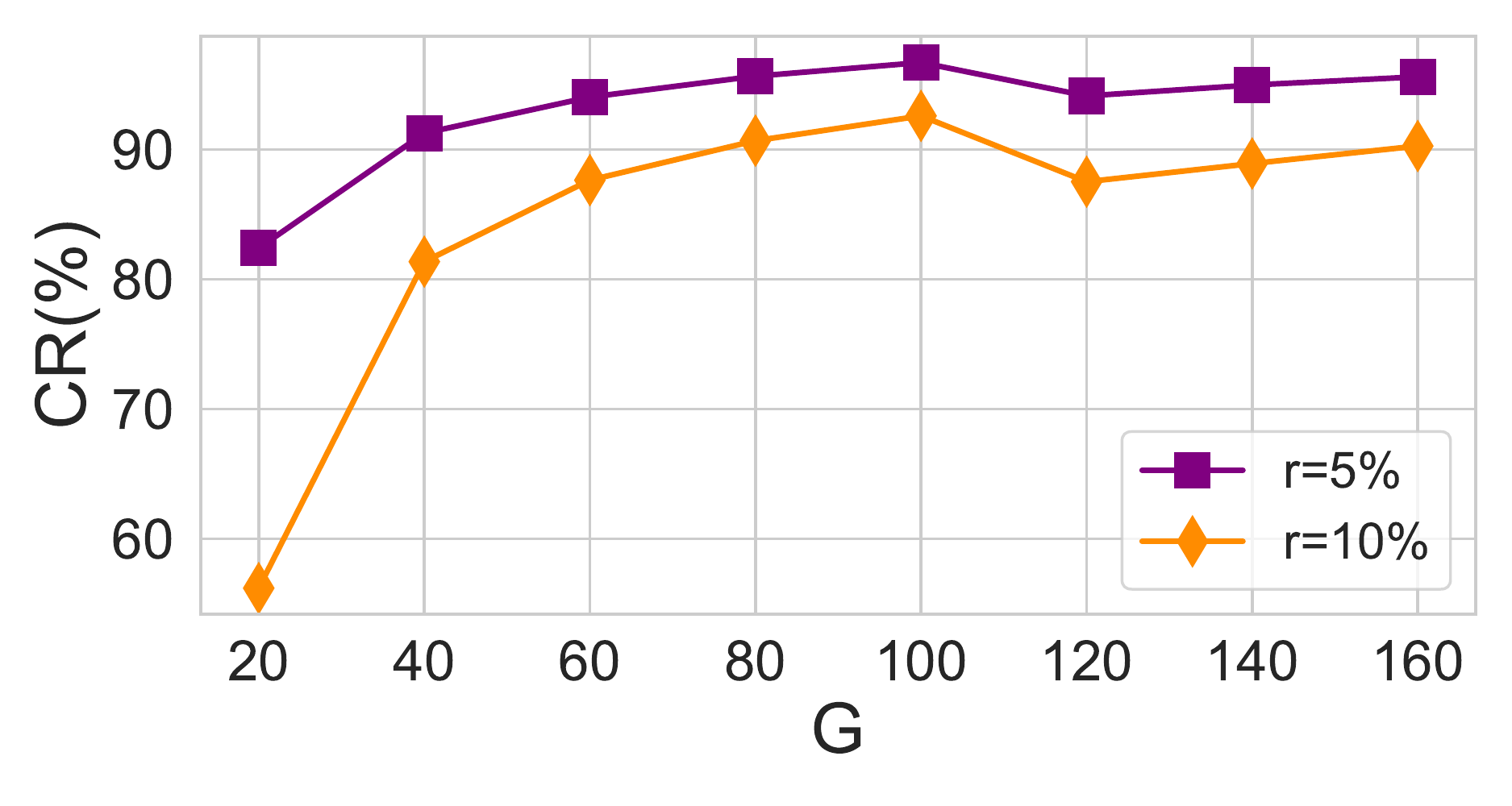}
    \caption{CR v.s. $G$}
    \label{fig:ablate_g}
  \end{subfigure}%
  \begin{subfigure}[b]{0.25\textwidth}
  \centering
    \includegraphics[width=1.\linewidth]{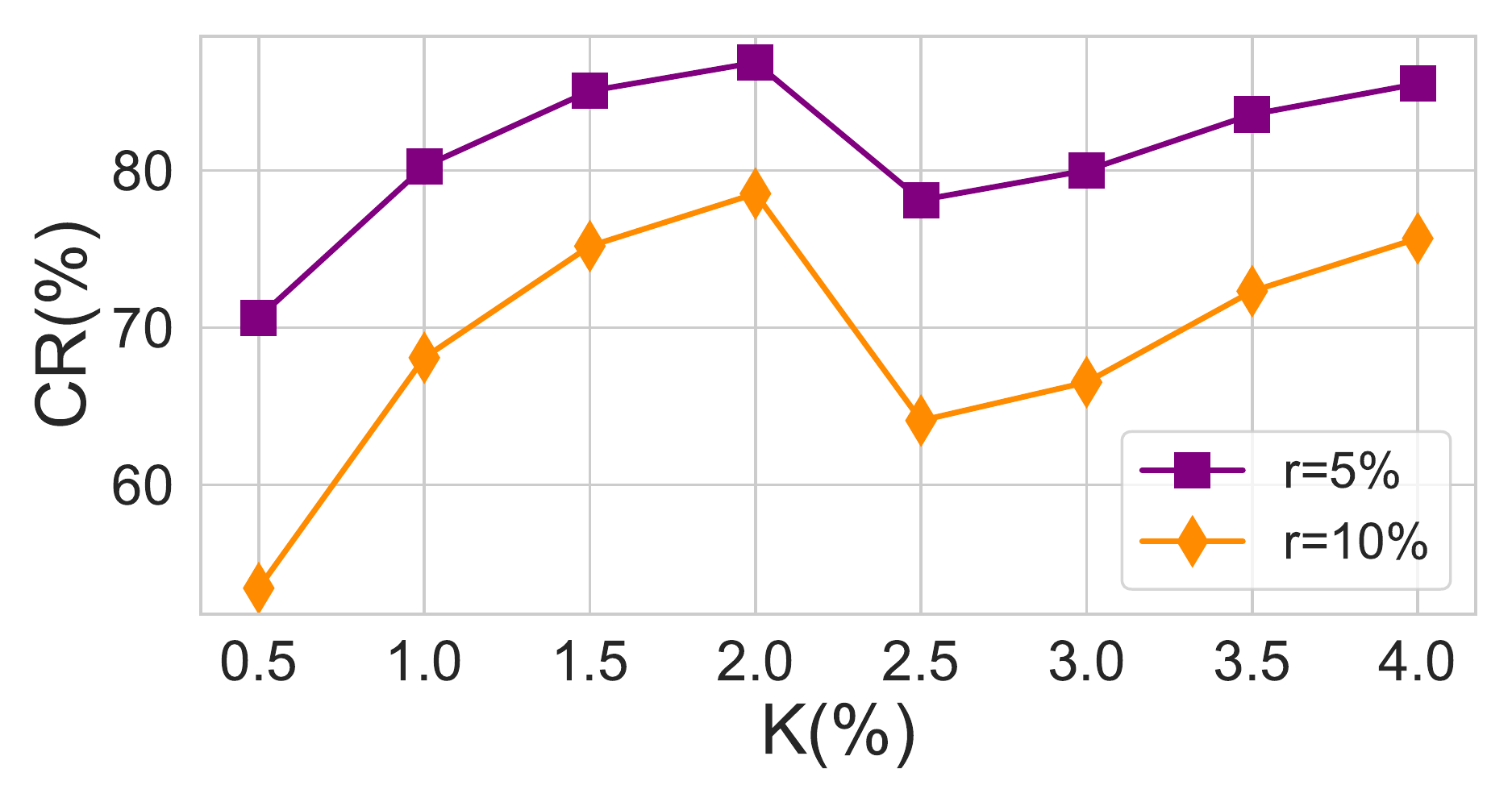}
    \caption{CR v.s. $K$}
    \label{fig:ablate_k}
  \end{subfigure}
  \\
  \begin{subfigure}[b]{0.4\textwidth}
  \centering
    \includegraphics[width=.85\linewidth]{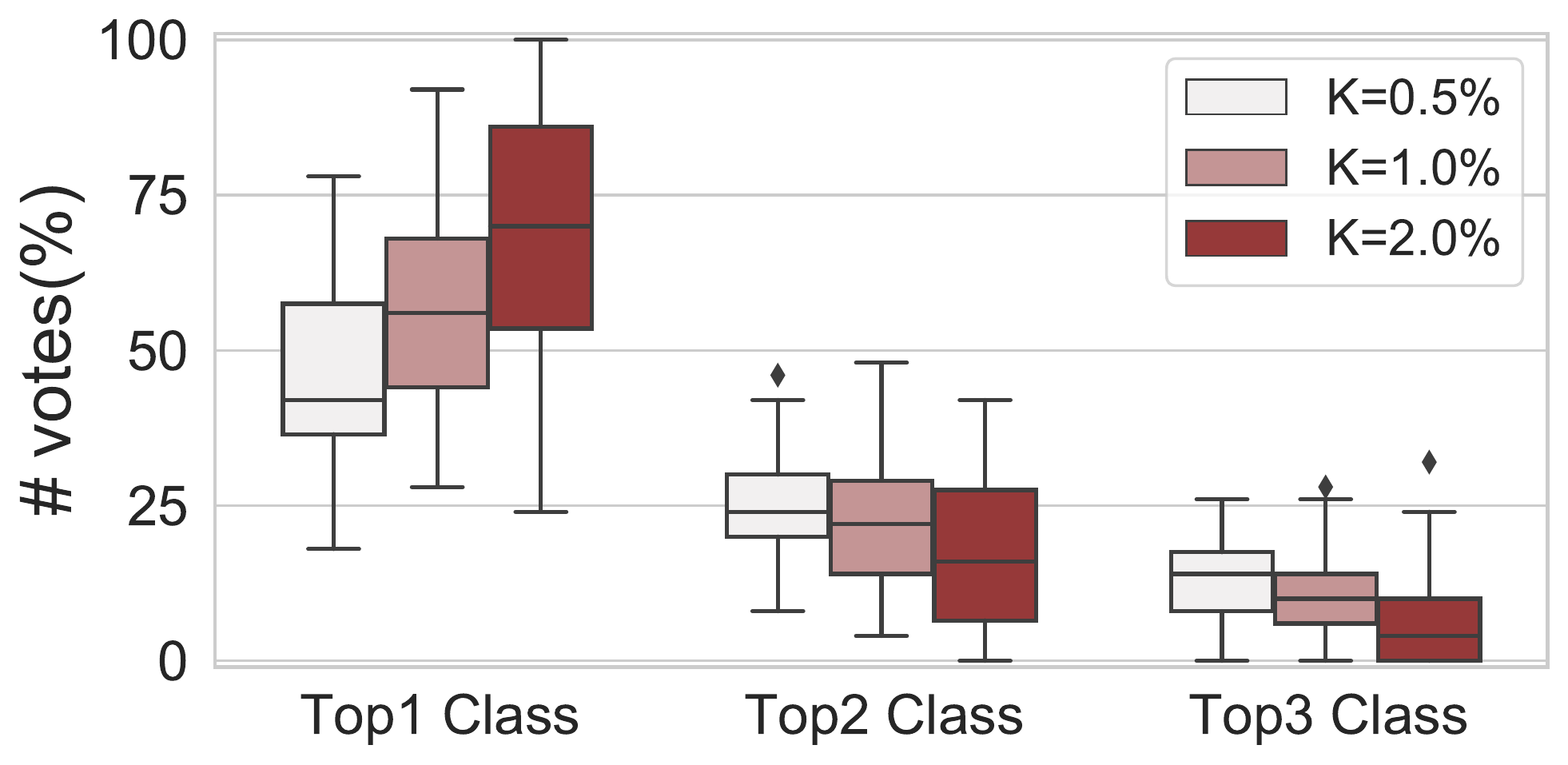}
    \caption{Sub-classifier voting distribution}
    \label{fig:votes_k}
  \end{subfigure}%
  \\
  \begin{subfigure}[b]{0.25 \textwidth}
  \centering
    \includegraphics[width=1.\linewidth]{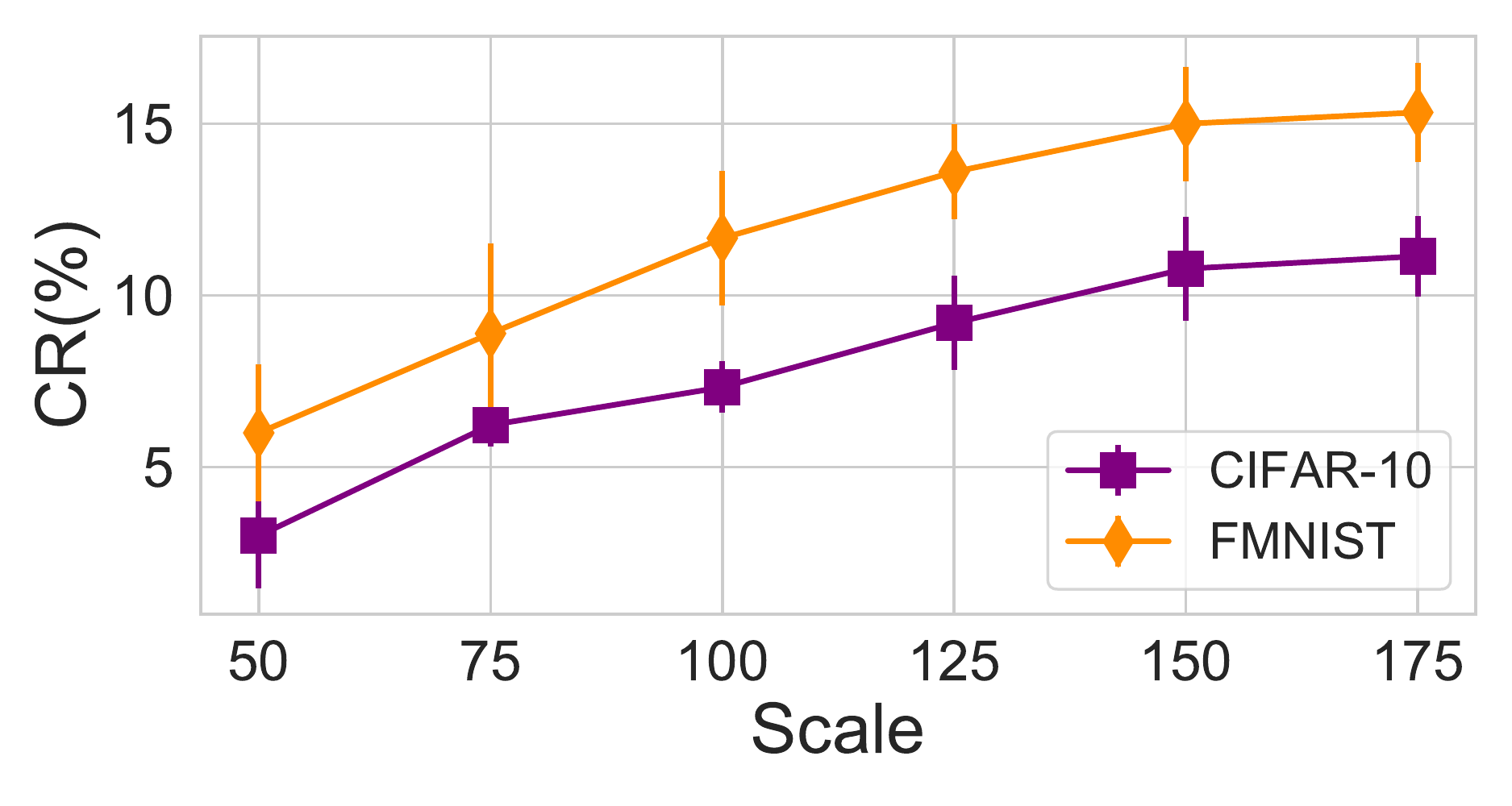}
    \caption{CR v.s. $\Delta$ (no time limit)}
    \label{fig:ablate_d1}
  \end{subfigure}%
  \begin{subfigure}[b]{0.25\textwidth}
  \centering
    \includegraphics[width=1.\linewidth]{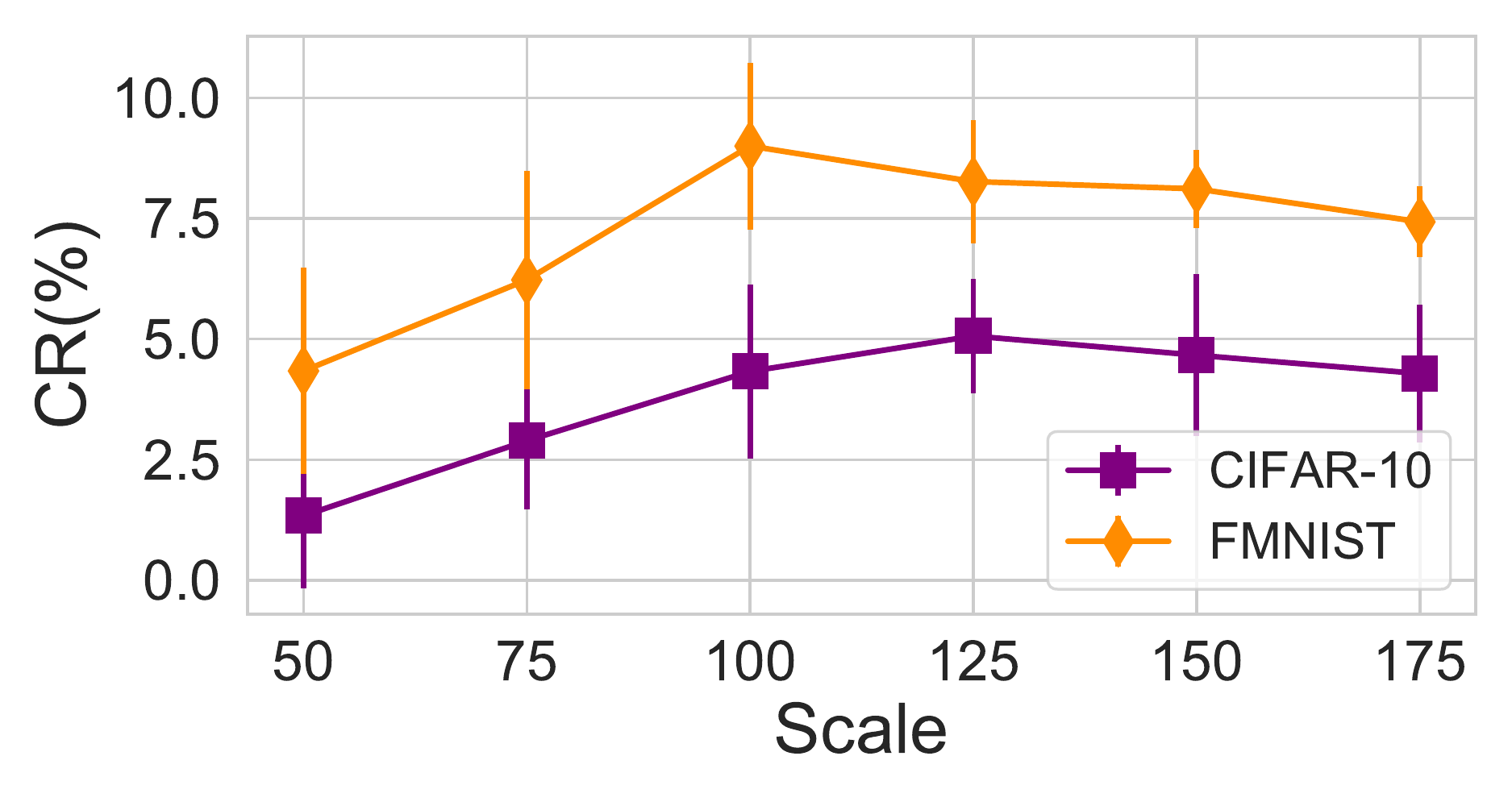}
    \caption{CR v.s. $\Delta$ (2s per sample)}
    \label{fig:ablate_d2}
  \end{subfigure}
\caption{Ablation study results on CV datasets. (a): $K=1\%N$ on FMNIST. (b): $G=50$ on CIFAR-10. (c): $G=50$ on CIFAR-10. (d) (e): $G=50$, $K=2\% N$, $r=30\% G$.} 
\end{figure}

\section{Conclusion}
Bagging, as a widely-used ensemble learning protocol, owns the certified robustness against data poisoning. In this paper, we derive the tight collective robustness certificate against the global poisoning attack for bagging. Current sample-wise certification is a specific variant of our collective certification. We also propose decomposition to accelerate the solving process. We analyze the upper bound of tolerable poison budget for vanilla bagging. Based on the analysis, we propose hash bagging to improve the certified robustness almost for free. Empirical results show the effectiveness of both our devised collective certification as well as the hash bagging. Our empirical results validate that: i) hash bagging is much robuster; ii) collective certification can yield a stronger collective robustness certificate.

\section*{Acknowledgements}
This work has been partially supported by the National Key R\&D Program of China No. 2020YFB1806700, NSFC Grant 61932014, NSFC Grant 61972246, Project BE2020026 supported by the Key R\&D Program of Jiangsu, China, and Shanghai Municipal Science and Technology Major Project (2021SHZDZX0102).

\bibliography{ref}
\bibliographystyle{icml2022}

\clearpage

\appendix

\section{Significance of Collective Robustness}
The fundamental difference between collective robustness and sample-wise robustness lies in \emph{the setting about the attacker objective}. For sample-wise robustness, the attacker aims to change a single prediction, while for collective robustness, the attacker aims to degrade the overall accuracy of a collection of predictions. Most data poisoning works ~\cite{wang2018data,goldblum2022dataset,geiping2020witches,huang2020metapoison,shafahi2018poison,wang2022improved} adopt the latter setting, which aim to maximize the attack success rate (the only metric in Poisoning Benchmark ~\cite{schwarzschild2021just}), hinting that collective robustness is more practical. In fact, sample-wise robustness is a special case of collective robustness when the collection size $M$=1, meaning that collective robustness is more general. In practice, if the model predicts a large collection of images at once, $M$ can be the collection size. If the model intermittently predicts a few images, $M$ can be the total number of the history predictions.

\section{Proofs}
\subsection{Proof of Prop.~\ref{pro:bagging}}

\begin{proposition}[Collective robustness of vanilla bagging]
For testset $\mathcal{D}_{test}=\{x_j\}_{j=0}^{M-1}$, we denote $\hat{y}_j= g(x_j)$ ($j=0,\ldots,M-1$) the original ensemble prediction, and $\mathcal{S}_i=\{g \mid s_i \in \mathcal{D}_g\}$ the set of the indices of the sub-trainsets that contain $s_i$. Then, the maximum number of simultaneously changed predictions (denoted by $M_{\rm{ATK}}$) under $r_{\rm{mod}}$ adversarial modifications, is computed by $\mathbf{(P1)}$:
\begin{align} 
    &\mathbf{(P1):} \quad M_{\rm{ATK}}= \max_{P_0,\ldots,P_{N-1}} \; \sum_{x_j \in \mathcal{D}_{test}} \mathbb{I} \left\{ \ovo_{x_j}(\hat{y}_j) < \right. \nonumber\\ 
    & \left. \qquad \qquad \qquad \qquad \max_{y \neq \hat{y}_j} \left[\ovo_{x_j}(y)+ \frac{1}{2}\mathbb{I}\{y <\hat{y}_j\}\right] \right\} \label{eq1:obj}\\
     & s.t. \quad [P_0, P_1, \ldots, P_{N-1}] \in \{0,1\}^N \label{eq1:c0}\\
     & \qquad \sum_{i=0}^{N-1} P_i \leq r_{\rm{mod}} \label{eq1:c1}\\
     & \ovo_{x_j}(\hat{y}_j) = \underbrace{\vo_{x_j}(\hat{y}_j)}_{\textbf{Original votes}}- \underbrace{\sum_{g=0}^{G-1} \mathbb{I}\{g \in \bigcup_{\forall i,P_i=1}\mathcal{S}_i\} \mathbb{I}\{ f_g(x_j)=\hat{y}_j \}  }_{\textbf{Influenced votes}} \nonumber \\
     & \qquad  \qquad \forall x_j \in \mathcal{D}_{test}, \; \hat{y}_j=g(x_j) \; \label{eq1:c2}\\
     & \ovo_{x_j}(y)= \underbrace{\vo_{x_j}(y)}_{\textbf{Original votes}}+ \underbrace{\sum_{g=0}^{G-1} \mathbb{I}\{ g \in \bigcup_{\forall i, P_i=1}\mathcal{S}_i\} \mathbb{I}\{ f_g(x_j) \neq y \} }_{\textbf{Influenced votes}} \nonumber \\
     & \qquad \qquad \forall x_j \in \mathcal{D}_{test},\; \forall y \in \Y, y\neq \hat{y}_j \label{eq1:c3} 
\end{align}
The collective robustness of vanilla bagging is $M-M_{\rm{ATK}}$.
\end{proposition}

\begin{proof}
The collective robustness is defined as the minimum number of simultaneously unchanged predictions, which is equal to the total number of predictions $M$ minus the maximum number of simultaneously changed predictions (denoted as $M_{\rm{ATK}}$). To compute the collective robustness, we only need to compute $M_{\rm{ATK}}$. $M_{\rm{ATK}}$ equals the objective value of:
\begin{align}
    & \max_{P_0,\ldots,P_{N-1}} \; \sum_{x_j \in \mathcal{D}_{test}} \mathbb{I} \{ \ovo_{x_j}(\hat{y}_j) \nonumber\\
    & < \max_{y \neq \hat{y}_j} \left[\ovo_{x_j}(y)+ \frac{1}{2}\mathbb{I}\{y <\hat{y}_j\}\right] \}
\end{align}
where $\ovo_{x_j}(y)$ denotes the number of votes for class $y$ when predicting $x_j$, after being attacked. We now explain each equation. Eq.~\ref{eq1:obj}: for the prediction of $x_j$, the prediction is changed only if \emph{there exists a class that obtains more votes than $y_j$ or the same number of votes but with a smaller index}. We consider three cases for the prediction of $x_j$:\\
\textbf{Case I:} $\ovo_{x_j}(\hat{y}_j) < \max_{y \neq \hat{y}_j}\ovo_{x_j}(y)$: we have $\ovo_{x_j}(\hat{y}_j) < \max_{y \neq \hat{y}_j}\ovo_{x_j}(y)+ \frac{1}{2}\mathbb{I}\{y <\hat{y}_j\}$, and the prediction of $x_j$ is changed.\\
\textbf{Case II:} $\ovo_{x_j}(\hat{y}_j) = \max_{y \neq \hat{y}_j}\ovo_{x_j}(y)$: whether the prediction is changed is determined by $\mathbb{I}\{y <\hat{y}_j\}$. If $\mathbb{I}\{y <\hat{y}_j\}=0$, meaning that there is no majority class with the smaller index than $\hat{y}_j$, then the prediction $\hat{y}_j$ is unchanged. Otherwise the prediction is changed.\\
\textbf{Case III:} $\ovo_{x_j}(\hat{y}_j) > \max_{y \neq \hat{y}_j}\ovo_{x_j}(y)$: we have $\ovo_{x_j}(\hat{y}_j) > \max_{y \neq \hat{y}_j}\ovo_{x_j}(y)+ \frac{1}{2}\mathbb{I}\{y <\hat{y}_j\}$, and the prediction of $x_j$ is unchanged.\\
We model the attack as $[P_0, P_1, \ldots, P_{N-1}] \in \{0,1\}^N$ where $P_i=1$ means that the attacker modifies the $i$-th training sample $s_i$. Since the attacker is only allowed to modify $r_{\rm{mod}}$ samples, we bound $\sum_{i=0}^{N-1} P_i \leq r_{\rm{mod}}$. We consider the predictions from the sub-classifiers whose sub-trainsets are changed, as the influenced predictions. Those influenced predictions are considered to be fully controlled by the attacker under our threat model. For the fixed $[P_0, P_1, \ldots, P_{N-1}]$, to maximize the number of simultaneously changed predictions, the optimal strategy is to change all the influenced predictions that equals $\hat{y}_j$ to other classes. Thus we have 
\begin{equation} \footnotesize
    \ovo_{x_j}(\hat{y}_j) = \underbrace{\vo_{x_j}(\hat{y}_j)}_{\textbf{Original votes}}- \underbrace{\sum_{g=0}^{G-1} \mathbb{I}\{g \in \bigcup_{\forall i,P_i=1}\mathcal{S}_i\} \mathbb{I}\{ f_g(x_j)=\hat{y}_j \}  }_{\textbf{Influenced votes}}
\end{equation}
Note that the attacker can arbitrarily manipulate the influenced predictions, so the number of votes for $y \neq y_j$ is 
\begin{equation} \footnotesize
    \ovo_{x_j}(y)= \underbrace{\vo_{x_j}(y)}_{\textbf{Original votes}}+ \underbrace{\sum_{g=0}^{G-1} \mathbb{I}\{ g \in \bigcup_{\forall i,P_i=1}\mathcal{S}_i\} \mathbb{I}\{ f_g(x_j) \neq y \} }_{\textbf{Influenced votes}}
\end{equation}
\textbf{Tightness.} The collective robustness $M- M_{\rm{ATK}}$ is tight for: 1) if the computed collective robustness $M- M_{\rm{ATK}}$ is lower than the actual collective robustness, meaning that our computed $M_{\rm{ATK}}$ is higher than the maximum number of simultaneously changed predictions, which contradicts the fact that we have find an attack that can achieve $M_{\rm{ATK}}$ under our threat model. 2) if the computed collective robustness $M- M_{\rm{ATK}}$ is higher than the actual collective robustness, meaning that our computed $M_{\rm{ATK}}$ is lower than the maximum number of simultaneously changed predictions, which contradicts the fact that $M_{\rm{ATK}}$ is the optimal objective value under our threat model.
\end{proof}

\subsection{Proof of NP-hardness}
\label{sec:NPhard}
We reformulate $\mathbf{(P1)}$ into the standard form of a BILP problem, which has been shown to be an NP-Complete problem~\cite{BILP}, to prove its NP-hardness. 
\begin{proof}
First of all, we introduce four sets of binary variables:
\begin{equation}
    \begin{aligned}
    \mathbf A &= [A_0,A_1,\dots,A_i,\dots,A_{G-1}] \in \{0,1\}^G, \\
    \mathbf Y &= [Y_0,Y_1,\dots,Y_j,\dots,Y_{M-1}] \in \{0,1\}^M, \\
    \mathbf Z &= [Z_{0,0},Z_{0,1},\dots,Z_{j,l},\dots,Z_{M-1,C-1}] \in \{0,1\}^{M \times C}, \\
    \mathbf W &= [W_0,W_1,\dots,W_k,\dots,W_{N-1}] \in \{0,1\}^N,
    \end{aligned}
\end{equation}
where $\mathbf A$ denotes the selected sub-classifiers to attack, $\mathbf Y$ denotes the attacked test samples, $\mathbf Z$ is an auxiliary set of binary variables for the prediction classes, $\mathbf W$ represents the poisoned training samples. In according with the main text, $G$ is the number of sub-classifiers, $M$ denotes the number of test samples, $C$ is the number of prediction classes, $N$ represents the number of training samples. 

With the notations defined above, we can reformulate $\mathbf{(P1)}$ as follows:
\begin{align}
    \mbox{Maximize}\qquad\qquad\qquad\quad M_{ATK} &= \sum_{j=0}^{M-1} Y_j \label{eq:newobj} \\
    s.t.\qquad\qquad\qquad\qquad\quad \sum_{k=0}^{N-1} &W_k \leq r_{mod}\label{eq:newcon1} \\
    \forall i, \;  A_i \leq \sum_{k=1}^{N-1} &W_k\mathbb I\{i \in \mathcal{S}_k\}\label{eq:newcon2} \\
    \forall j, l\neq \hat{y}_j,i,\ \mbox{either}\ Z_{j,l} \leq 0\ \mbox{or}\ &V_{x_j}(\hat{y}_j) - V_{x_j}(l) \leq \nonumber \\  
    \sum_{i=0}^{G-1} A_i (\mathbb I\{f_i(x_j) \neq l\} + \mathbb I &\{f_i(x_j)=\hat{y}_j\})\label{eq:newcon3} \\
    \forall j,\ \mbox{either}\ Y_j \leq 0\ &\mbox{or}\ \sum_{l=0}^{C-1}Z_{j,l}\geq 2 \label{eq:newcon4}
\end{align}

We now explain each equation respectively. Eq.~(\ref{eq:newobj}) is the variant of Eq.~(\ref{eq1:obj}), denoting that our objective is to maximize the number of attacked test samples. Eq.~(\ref{eq:newcon1}) shares the same meaning as Eq.~(\ref{eq1:c1}), which restricts the number of poisoned training samples to be less than $r_{mod}$. Eq.~(\ref{eq:newcon2}) restricts the selected sub-classifiers should be in $\bigcup_{\forall k,P_k=1}\mathcal S_k$. Eq.~(\ref{eq:newcon3}) shows that $Z_{j,l}$ could be 1 only when the ensemble prediction of the test sample $j$ can be changed from $\hat{y}_j$ to $l$ (we ignore the minimum index constraint for simplicity). Eq.~(\ref{eq:newcon4}) shows that $Y_j$ could be 1 (the test sample $j$ is attacked successfully) only when there exists some classes that the ensemble prediction can be changed to. We use the equation $\sum_{l=0}^{C-1}Z_{i,l}\geq 2$ since we always have $Z_{j,\hat{y}_j}=1$. 

The formulation above has been in the standard form of a BILP problem, except the ``either...or...'' clause. Using the transformation trick in~\cite{BILP}, e.g.
\begin{equation*}
    \mbox{either}\quad  x_1 + x_2 \leq 4 \quad  \mbox{or}\quad \ x_1 + 1.5x_2 \leq 6 
\end{equation*}
is equal to
\begin{align*}
    x_1 + x_2 &\leq 4 + My \\
    x_1 + 1.5x_2 &\leq 6 + M(1-y) 
\end{align*}
where $M$ is a large number, $y$ is an auxiliary introduced binary variable. 

Thus, we can transform Eq.~(\ref{eq:newcon3}) and Eq.~(\ref{eq:newcon4}) into the standard form of constraints by introducing additionally number and binary variables, which means that $\mathbf{(P1)}$ can be transformed into the standard form of a BILP problem. Now we can tell that $\mathbf{(P1)}$ is an NP-hard problem.
\end{proof}

\subsection{Proof of Prop.~\ref{pro:upper}}
\begin{proposition}[Upper bound of tolerable poison budget]
Given $\mathcal{S}_i$ ($i=0,\ldots,N-1$), \emph{the upper bound of the tolerable poisoned samples} (denoted by $\overline{r}$) is \vspace{-5pt}
\begin{equation} \vspace{-5pt}
    \overline{r}=\min |\Pi| \; s.t.\; |\bigcup_{i \in \Pi} \mathcal{S}_i| > G/2
\end{equation}
which equals the minimum number of training samples that can influence more than a half of sub-classifiers.
\end{proposition}
\begin{proof}
We prove that, $\forall r_{\rm{mod}} \geq \overline{r}$, the collective robustness computed from $\mathbf{(P1)}$ is $0$. Specifically, when $r_{\rm{mod}} \geq \overline{r}$, if we choose to poison the training samples whose indices are within $\Pi$, for all $\hat{y}_j$, the number of votes for the original ensemble prediction $\hat{y}_j$ is
\begin{footnotesize}
\begin{align}
    & \ovo_{x_j}(\hat{y}_j) =\vo_{x_j}(\hat{y}_j)- \sum_{g=0}^{G-1} \mathbb{I}\{g \in \bigcup_{\forall i,P_i=1}\mathcal{S}_i\} \mathbb{I}\{ f_g(x_j)=\hat{y}_j \} \\
    &=\vo_{x_j}(\hat{y}_j)-\sum_{g=0}^{G-1} \mathbb{I}\{g \in \bigcup_{i \in \Pi}\mathcal{S}_i\} \mathbb{I}\{ f_g(x_j)=\hat{y}_j \}\\
    & = \sum_{g=0}^{G-1} \mathbb{I}\{ f_g(x_j)=\hat{y}_j \}- \sum_{g=0}^{G-1} \mathbb{I}\{g \in \bigcup_{i \in \Pi}\mathcal{S}_i\} \mathbb{I}\{ f_g(x_j)=\hat{y}_j \} \\
    & \leq \sum_{g=0}^{G-1} \mathbb{I}\{g \not\in \bigcup_{i \in \Pi}\mathcal{S}_i\} \\
    & < \frac{G}{2}
\end{align}
\end{footnotesize}
The number of votes for other classes $y \neq \hat{y}_j$ is
\begin{footnotesize}
\begin{align}
    & \ovo_{x_j}(y)= \vo_{x_j}(y)+ \sum_{g=0}^{G-1} \mathbb{I}\{ g \in \bigcup_{\forall i,P_i=1}\mathcal{S}_i\} \mathbb{I}\{ f_g(x_j) \neq y \} \\
    &= \vo_{x_j}(y) + \sum_{g=0}^{G-1} \mathbb{I}\{g \in \bigcup_{i \in \Pi}\mathcal{S}_i\} \mathbb{I}\{ f_g(x_j)  \neq y \}\\
    & \geq \sum_{g=0}^{G-1} \mathbb{I}\{g \in \bigcup_{i \in \Pi}\mathcal{S}_i\}\\
    & > \frac{G}{2}
\end{align}
\end{footnotesize}
We have
\begin{footnotesize}
\begin{align}
    & \ovo_{x_j}(\hat{y}_j) - \max_{y \neq \hat{y}_j} \left[\ovo_{x_j}(y)+ \frac{1}{2}\mathbb{I}\{y <\hat{y}_j\} \right] \\
    & \leq \frac{G}{2}-\frac{G}{2}+1 -\frac{1}{2}\\
    & <0
\end{align}
\end{footnotesize}
Therefore, $\forall x_j$, the prediction $\hat{y}_j$ is considered to be corrupted. The certified collective robustness is $0$.
\end{proof}

\subsection{Proof of Prop.~\ref{pro:hash}}
\begin{proposition}[Certified collective robustness of hash bagging]
For testset $\mathcal{D}_{test}=\{x_j\}_{j=0}^{M-1}$, we denote $\hat{y}_j= g(x_j)$ ($j=0,\ldots,M-1$) the ensemble prediction. The maximum number of simultaneously changed predictions (denoted by $M_{\rm{ATK}}$) under $r_{\rm{ins}}$ insertions, $r_{\rm{del}}$ deletions and $r_{\rm{mod}}$ modifications, is computed by $\mathbf{(P2)}$:
\begin{footnotesize}
\begin{align} 
    &\mathbf{(P2)}: \quad M_{\rm{ATK}}= \max_{A_0,\ldots,A_{G-1}} \; \sum_{x_j \in \mathcal{D}_{test}} \mathbb{I} \left\{ \ovo_{x_j}(\hat{y}_j) < \right. \nonumber\\ 
    & \left. \qquad \qquad \qquad \qquad \max_{y \neq \hat{y}_j} \left[\ovo_{x_j}(y)+ \frac{1}{2}\mathbb{I}\{y <\hat{y}_j\}\right] \right\} \label{eq1:p2obj}\\
     & s.t. \quad [A_0, A_1, \ldots, A_{G-1}] \in \{0,1\}^G \label{eq1:p2c0}\\
     & \sum_{g= (l-1) \hat{G}}^{l \hat{G} -1} A_g \leq r_{\rm{ins}}+r_{\rm{del}}+2r_{\rm{mod}} \nonumber \\
     & \qquad \qquad l=1,\ldots, \lceil G/\hat{G} \rceil \label{eq1:p2c1}\\
     & \ovo_{x_j}(\hat{y}_j) = \underbrace{\vo_{x_j}(\hat{y}_j)}_{\textbf{Original votes}}- \underbrace{\sum_{g=1}^G A_g \mathbb{I}\{ f_g(x_j)=\hat{y}_j \}}_{\textbf{Influenced votes}}  \nonumber \\
     & \qquad \qquad \forall x_j \in \mathcal{D}_{test} \label{eq1:p2c2}\\
     & \ovo_{x_j}(y)= \underbrace{\vo_{x_j}(y)}_{\textbf{Original votes}}+ \underbrace{\sum_{g=1}^G A_g \mathbb{I}\{ f_g(x_j) \neq y \}}_{\textbf{Influenced votes}} \nonumber \\
     & \qquad \qquad \forall x_j \in \mathcal{D}_{test},\; \forall y \neq \hat{y}_j \label{eq1:p2c3}
\end{align}
\end{footnotesize}
The collective robustness is $M-M_{\rm{ATK}}$.
\end{proposition}

\begin{proof}
In fact, $\mathbf{(P2)}$ is a simplified version of $\mathbf{(P1)}$ which exploits the properties of hash bagging. $\mathbf{(P2)}$ is mainly different from $\mathbf{(P1)}$ in Eq. (\ref{eq1:c0}) and Eq. (\ref{eq1:c1}). Specifically, in $\mathbf{(P2)}$, the poisoning attack is expressed as $[A_0, A_1, \ldots, A_{G-1}]$, where $A_g$ denotes whether the $g$-th sub-classifier is influenced, instead of whether the $g$-th sample is modified in $\mathbf{(P1)}$. Based on the property of hash bagging, each trainset-hash pair $(\mathcal{D}_{train}, \textrm{Hash}(\cdot))$ is partitioned into $\lfloor N/K \rfloor$ disjoint sub-trainsets. Therefore, $r_{\rm{ins}}$ insertions, $r_{\rm{del}}$ deletions and $r_{\rm{mod}}$ modifications can influence at most $r_{\rm{ins}}+r_{\rm{del}}+2r_{\rm{mod}}$ sub-trainsets within each trainset-hash pair, as shown in Eq. (\ref{eq1:p2c1}).

\textbf{Tightness.} When $N \leq GK$, the proof of tightness is the same as that for $\mathbf{(P1)}$. Next, we prove that our robustness is tight. In particular, we prove: i) the collective robustness computed from $\mathbf{(P2)}$ is a lower bound. ii) the collective robustness $M-M_{\rm{ATK}}$ by $\mathbf{(P2)}$ is an upper bound.

\textbf{i)} For arbitrary $r_{\rm{ins}}$ insertions, $r_{\rm{del}}$ deletions and $r_{\rm{mod}}$ modifications can influence at most $r_{\rm{ins}}+r_{\rm{del}}+2r_{\rm{mod}}$ sub-trainsets within each trainset-hash pair. Therefore, for any poisoning attack ($r_{\rm{ins}}$ insertions, $r_{\rm{del}}$ deletions and $r_{\rm{mod}}$ modifications), we can denote it by $[A_0, A_1, \ldots, A_{G-1}]$:
\begin{align}
    & [A_0, A_1, \ldots, A_{G-1}] \in \{0,1\}^G \nonumber \\
     & \sum_{g= (l-1) \hat{G}}^{l \hat{G} -1} A_g \leq r_{\rm{ins}}+r_{\rm{del}}+2r_{\rm{mod}} \nonumber
\end{align}
The poisoning attacks denoted by Eq. (\ref{eq1:p2c0}), Eq. (\ref{eq1:p2c1}) are stronger than the practical poisoning attacks. Therefore, the collective robustness computed from $\mathbf{(P2)}$ is a lower bound.

\textbf{ii)} First we denote $\{A_{(l-1)\hat{G} + \beta_{l,o}} \mid o=0,\ldots,r-1; \; l=1,\ldots, \lceil G/\hat{G} \rceil;\; \beta_{l,o} \in [0, \hat{G}-1] \}$ the influenced sub-classifiers ($A_{(l-1)\hat{G} + \beta_{l,o}}=1$). We construct an insertion attack as follow: we insert $r$ new samples (denoted by $\hat{s}_o: \; o=0,\dots,r-1$), where the hash value of $\hat{s}_o$ computed by the $l$-th hash function mod $\hat{G}$ is  $\beta_{l,o}$. We can achieve $M_{\rm{ATK}}$ within poison budget $r$. Therefore, the collective robustness $M-M_{\rm{ATK}}$ is an upper bound.
\end{proof}

\section{Certification Gap}
We intuitively show the gap between the collective robustness guaranteed by our collective certification and that of the sample-wise certification in Fig.~\ref{fig:gap}. 

\section{Comparison Overview}
Table~\ref{tab:compare} presents an overview of the theoretical comparisons to other certified defenses that are tailored to the general data poisoning attack.

\begin{table}[tb!]
	\centering
	\caption{Method comparison. \textbf{Model, Training, Bagging} denote whether the defense is compatible with various classifier models, training algorithms and general forms of bagging. \textbf{Sample-wise, Collective, Deterministic} denote whether the method can provide sample-wise robustness certificates, collective robustness certificates and deterministic robustness certificates. }
	\label{tab:compare}
	\begin{adjustbox}{width=0.98\linewidth}{
	\begin{threeparttable}
    {{	\begin{tabular}{c|lcccccccc}
			\toprule[0.8pt]
			& \multirow{2}{*}{Methods}  & \multicolumn{3}{|c|}{Certified Defense} & \multicolumn{3}{c}{Robustness Certification}\\
			\cline{3-8}
			& & \multicolumn{1}{|c}{Model} & Training & Bagging & \multicolumn{1}{|c}{Sample-wise} & Collective & Deterministic\\
			\midrule[0.6pt]
          \multirow{6}{*}{} 
            &~\cite{levine2021deep}  & \checkmark & \checkmark & \XSolid  & \checkmark & \XSolid & \checkmark\\
            &~\cite{jia2021intrinsic}   & \checkmark & \checkmark & \XSolid & \checkmark & \XSolid & \XSolid\\
            &~\cite{ma2019data} & \checkmark & \XSolid & -- & \checkmark & \XSolid & \XSolid \\
            &~\cite{jia2022rnn} & \XSolid & \XSolid & -- & \checkmark & \checkmark & \checkmark\\
            \midrule[0.5pt]
            & Ours  & \checkmark & \checkmark & \checkmark & \checkmark & \checkmark & \checkmark \\
			\bottomrule[0.8pt]
		\end{tabular}
		}}
	\end{threeparttable}
}
	\end{adjustbox}
\end{table}

\begin{figure}[t!]
  \centering
    \includegraphics[width=0.68\linewidth]{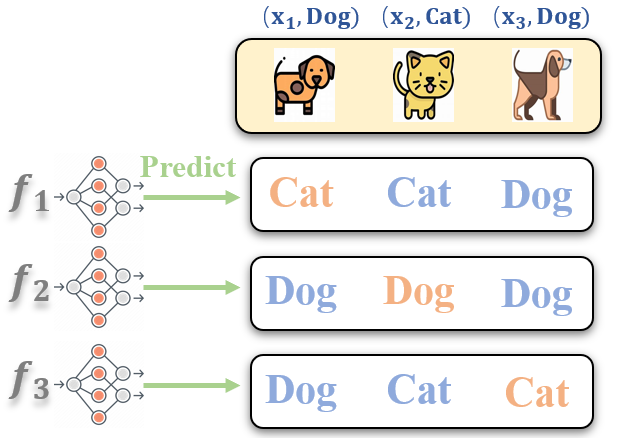}
    \caption{An example to illustrate the gap between the sample-wise certificate and the collective certificate. Suppose the sub-classifiers are $f_1(x), f_2(x), f_3(x)$, and the testing samples are $x_1,x_2,x_3$. The predictions \textcolor{lightblue}{Cat/Dog} are correct, and \textcolor{lightred}{Cat/Dog} are wrong. Consider an attacker (poison budget is 1) can control an arbitrary sub-classifier. \textbf{Sample-wise certificate}: we consider $g(x_1), g(x_2), g(x_3)$ independently. To change $g(x_1)$/$g(x_2)$/$g(x_3)$, the attacker can flip $f_2(x_1)$/$f_3(x_2)$/$f_1(x_3)$ respectively. Therefore, all the three predictions are not robust and the sample-wise robustness is $0$. \textbf{Collective certificate}: we consider $g(x_1), g(x_2), g(x_3)$ collectively. If the attacker poisons $f_1$/$f_2$/$f_3$, the prediction $g(x_1)$/$g(x_2)$/$g(x_3)$ is unchangeable respectively. Thus the collective robustness is $1$.}
  \label{fig:gap}
\end{figure}

\section{Implementation Details}
\label{sec:implement}
All the experiments are conducted on CPU (16 Intel(R) Xeon(R) Gold 5222 CPU @ 3.80GHz) and GPU (one NVIDIA RTX 2080 Ti).

\subsection{Training Algorithm}
Alg.~\ref{alg:train} summarizes our training process for hash bagging. It needs to set the random seed for reproducible training and train the sub-classifiers on the hash-based sub-trainsets.

\begin{algorithm2e}[t]
  \caption{Train the sub-classifiers.}
  \label{alg:train}
  \KwIn{trainset $\mathcal{D}_{train}$, number of sub-trainsets $G$, sub-trainset size $K$, hash functions $\rm{Hash}_{h}(\cdot): h=0,1,\ldots$.}
   Construct $G$ sub-trainsets $\mathcal{D}_g$ ($g=0,\ldots,G-1$) based on Eq. (\ref{eq:sub-trainset2}); \algcomment{Hash-based subsampling.} \\
   Set the random seed for training; \algcomment{Reproducible training.} \\
   Train the sub-classifiers $f_g$ on $\mathcal{D}_g$ ($g=0,\ldots,G-1$); \\
  \KwOut{The trained sub-classifiers $\{f_g\}_{g=1}^G$.}
\end{algorithm2e}

\begin{table*}[tb!]
\centering
\caption{Impact of $\Delta$ ($K=N/G$). The numerical results record the mean and variance of the certified robustness ratio. NaN: The number of breakable test samples $M \leq 6|\Delta|$ so we cannot calculate valid variance for CR ratios.} \label{tab:delta}
\begin{adjustbox}{width=\linewidth}
    \begin{tabular}{lll | cccccccccccc}
    \toprule
   Dataset & G & $\Delta$ & $5\%$ & $10\%$ & $15\%$ & $20\%$ & $25\%$ & $30\%$ & $35\%$ & $40\%$ & $45\%$ & $50\%$\\
    \midrule
   \multirow{14}{*}{FMNIST} & \multirow{6}{*}{50} 
    & 50 & \textbf{13.00$\pm$ 2.76} &
  15.00$\pm$ 5.86 &
  15.00$\pm$ 5.98 &
  11.66$\pm$ 3.54 &
  6.34$\pm$ 3.54 &
  4.34 $\pm$ 2.14 &
  1.00$\pm$ 1.00 &
  0.66$\pm$ 0.94 &
  0.00$\pm$ 0.00 &
  0.00$\pm$ 0.00 \\
    & & 75 & NaN &
  \textbf{19.56$\pm$  3.97} &
  \textbf{18.22$\pm$  5.59} &
  \textbf{16.22$\pm$ 2.92} &
  10.89$\pm$ 3.88 &
  6.22$\pm$ 2.27 &
  4.67$\pm$ 1.84 &
  1.11$\pm$ 0.92 &
  0.00$\pm$ 0.00 &
  0.00$\pm$ 0.00 \\
    & & 100 & NaN &
  18.17$\pm$ 0.74 &
  15.50$\pm$ 1.71 &
  13.17$\pm$ 3.02 &
  \textbf{12.47$\pm$ 1.34} &
  9.00$\pm$ 1.73 &
  6.5$\pm$ 1.61 &
  3.17$\pm$ 1.34 &
  0.00$\pm$ 0.00 &
  0.00$\pm$ 0.00 \\
    & & 125 & NaN &
  NaN &
  12.00$\pm$ 1.37 &
  11.33$\pm$ 0.72 &
  10.8$\pm$ 1.10 &
  \textbf{8.26$\pm$ 1.28} &
  \textbf{7.2$\pm$ 1.53} &
  4.67$\pm$ 1.07 &
  0.00$\pm$ 0.00 &
  0.00$\pm$ 0.00\\
    & & 175 & NaN &
  NaN &
  NaN &
  9.61$\pm$ 1.01 &
  8.38$\pm$ 0.63 &
  7.43$\pm$ 0.74 &
  5.81$\pm$ 0.95 &
  \textbf{5.62$\pm$ 1.21} &
  0.38$\pm$ 0.42 &
  0.00$\pm$ 0.00 \\
    & & 200 & NaN &
  NaN &
  NaN &
  8.66$\pm$ 1.25 &
  8.08$\pm$ 0.67 &
  7.08$\pm$ 1.06 &
  5.66$\pm$ 1.18 &
  5.25$\pm$ 0.75 &
  \textbf{0.84$\pm$ 0.75} &
  0.00$\pm$ 0.00\\
    \cmidrule(l){2-13}
    & \multirow{6}{*}{100} 
    & 50 & \textbf{13.34$\pm$2.74} &
  \textbf{13.34$\pm$3.40} &
  8.00$\pm$5.04 &
  8.66$\pm$4.42 &
  4.00$\pm$3.26 &
  1.66$\pm$1.38 &
  2.00$\pm$2.30 &
  0.00$\pm$0.00 &
  0.00$\pm$0.00 &
  0.00$\pm$0.00 \\
    & & 100 & NaN & 11.50$\pm$1.71 &
  \textbf{10.34$\pm$1.70} &
  \textbf{10.00$\pm$1.41} &
  \textbf{7.84$\pm$2.03} &
  \textbf{5.50$\pm$3.0} &
  4.33$\pm$1.97 &
  1.00$\pm$1.15 &
  0.00$\pm$0.00 &
  0.00$\pm$0.00 \\
    & & 150 & NaN &
  NaN &
  7.89$\pm$1.46 &
  7.45$\pm$1.51 &
  5.45$\pm$1.18 &
  4.78$\pm$0.25 &
  \textbf{4.78$\pm$0.6} &
  2.45$\pm$0.99 &
  0.00$\pm$0.00 &
  0.00$\pm$0.00 \\
    & & 200 & NaN &
  NaN &
  6.25$\pm$0.56 &
  5.25$\pm$0.75 &
  4.50$\pm$1.08 &
  4.42$\pm$0.78 &
  3.50$\pm$0.81 &
  2.34$\pm$0.98 &
  0.42$\pm$0.34 &
  0.00$\pm$0.00 \\
    & & 250 & NaN &
  NaN &
  NaN &
  5.20$\pm$0.86 &
  4.27$\pm$0.72 &
  3.53$\pm$0.71 &
  3.47$\pm$0.79 &
  \textbf{2.47$\pm$1.07} &
  0.60$\pm$0.24 &
  0.00$\pm$0.00 \\
    & & 300 & NaN &
  NaN &
  NaN &
  NaN &
  4.00$\pm$0.58 &
  3.50$\pm$0.37 &
  2.44$\pm$0.85 &
  2.44$\pm$0.85 &
  \textbf{0.89$\pm$0.25} &
  0.00$\pm$0.00 \\
    \hline
    \multirow{14}{*}{CIFAR-10} & \multirow{6}{*}{50} 
    & 50 & 15.33$\pm$ 5.73 &
  10.33$\pm$ 2.43 &
  9.00$\pm$ 4.73 &
  7.67$\pm$ 2.13 &
  5.33$\pm$ 3.94 &
  1.33$\pm$ 1.49 &
  0.33$\pm$ 0.75 &
  0.00$\pm$ 0.00 &
  0.00$\pm$ 0.00 &
  0.00$\pm$ 0.00\\
    & & 75 & \textbf{17.56$\pm$ 0.92} &
  \textbf{11.56$\pm$ 2.73} &
  12.00$\pm$ 2.88 &
  \textbf{10.67$\pm$ 1.53} &
  7.78$\pm$ 2.23 &
  2.89$\pm$ 1.43 &
  0.22$\pm$ 0.49 &
  0.00$\pm$ 0.00 &
  0.00$\pm$ 0.00 &
  0.00$\pm$ 0.00\\
    & & 100 & 14.50$\pm$ 3.69 &
  10.33$\pm$ 0.74 &
  \textbf{12.00$\pm$ 1.41} &
  9.50$\pm$ 2.06 &
  \textbf{8.50$\pm$ 0.96} &
  4.33$\pm$ 1.80 &
  1.16$\pm$ 1.46 &
  0.00$\pm$ 0.00 &
  0.00$\pm$ 0.00 &
  0.00$\pm$ 0.00\\
    & & 125 & 11.87$\pm$ 1.56 &
  9.33$\pm$ 1.64 &
  10.00$\pm$ 1.37 &
  8.00$\pm$ 0.92 &
  7.73$\pm$ 0.88 &
  \textbf{5.07$\pm$ 1.19} &
  2.00$\pm$ 1.44 &
  0.80$\pm$ 0.80 &
  0.00$\pm$ 0.00 &
  0.00$\pm$ 0.00\\
    & & 175 & 10.00$\pm$ 1.83 &
  9.33$\pm$ 0.63 &
  7.24$\pm$ 1.13 &
  6.67$\pm$ 1.03 &
  5.9$\pm$ 0.63 &
  4.29$\pm$ 1.43 &
  \textbf{3.05$\pm$ 1.17} &
  1.14$\pm$ 0.74 &
  0.00$\pm$ 0.00 &
  0.00$\pm$ 0.00\\
    & & 200 & 8.17$\pm$ 3.41 &
  8.33$\pm$ 0.63 &
  7.17$\pm$ 0.94 &
  5.83$\pm$ 0.69 &
  5.33$\pm$ 0.47 &
  4.25$\pm$ 0.95 &
  \textbf{2.67$\pm$ 0.95} &
  \textbf{2.00$\pm$ 0.87} &
  0.00$\pm$ 0.00 &
  0.00$\pm$ 0.00\\
    \cmidrule(l){2-13}
    & \multirow{6}{*}{100} 
    & 50 & \textbf{11.00$\pm$ 3.42} &
  \textbf{9.66$\pm$3.54} &
  \textbf{5.66$\pm$4.82} &
  \textbf{3.66$\pm$2.42} &
  2.00$\pm$1.64 &
  0.66$\pm$0.94 &
  0.00$\pm$0.00 &
  0.00$\pm$0.00 &
  0.00$\pm$0.00 & 0.00 $\pm$ 0.00\\
    & & 100 & 7.67$\pm$ 2.56 &
  5.50$\pm$1.89 &
  5.33$\pm$2.21 &
  \textbf{5.00$\pm$1.82} &
  \textbf{4.50$\pm$2.14} &
  \textbf{2.50$\pm$0.96} &
  0.17$\pm$0.37 &
  0.00$\pm$0.00 & 0.00$\pm$0.00 & 0.00$\pm$0.00\\
    & & 150 & 7.11$\pm$ 1.25 &
  5.55$\pm$0.63 &
  4.22$\pm$0.49 &
  3.55$\pm$0.83 &
  2.11$\pm$0.46 &
  1.78$\pm$0.31 &
  0.89$\pm$0.49 &
  0.00$\pm$0.00 &
  0.00$\pm$0.00 & 
  0.00$\pm$0.00\\
    & & 200 &  5.34$\pm$ 2.32 &
  5.58$\pm$0.34 &
  4.34$\pm$0.80 &
  2.92$\pm$0.34 &
  2.75$\pm$0.48 &
  1.58$\pm$0.18 &
  1.00$\pm$0.50 &
  0.00$\pm$0.00 &
  0.00$\pm$0.00 & 
  0.00$\pm$0.00\\
    & & 250 & 3.93$\pm$ 2.51 &
  4.53$\pm$1.32 &
  4.13$\pm$0.72 &
  2.87$\pm$0.43 &
  2.20$\pm$0.30 &
  1.67$\pm$0.36 &
  \textbf{1.06$\pm$0.30} &
  \textbf{0.13$\pm$0.19} &
  0.00$\pm$0.00 &
  0.00$\pm$0.00\\
    & & 300 & 5.44$\pm$ 0.46 &
  4.61$\pm$0.65 &
  3.67$\pm$0.54 &
  2.78$\pm$0.31 &
  2.17$\pm$0.17 &
  1.56$\pm$0.16 &
  1.00$\pm$0.35 &
  0.06$\pm$0.12 &
  0.00$\pm$0.00 &
  0.00$\pm$0.00\\
    \bottomrule
    \end{tabular}
    \end{adjustbox}
\end{table*}

\subsection{Dataset Information}
Table~\ref{tab:config} shows our experimental setups in details.

\textbf{Bank\footnote{\url{https://archive.ics.uci.edu/ml/datasets/Bank+Marketing}.}} dataset consists of 45,211 instances of 17 attributes (including both numeric attributes and categorical attributes) in total. Each of the instances is labeled to two classes, ``yes'' or ``no''. We partition the dataset to 35,211 for training and 10,000 for testing. We use SVM as the sub-classifier architecture. 

\textbf{Electricity\footnote{\url{https://datahub.io/machine-learning/electricity}.}} has 45,312 instances of 8 numeric attributes. Each of the instances is labeled to two classes, ``up'' or ``down''. We partition the dataset to 35,312 for training and 10,000 for testing. Following ~\cite{newensemble}, we use Bayes as the sub-classifier architecture for ensemble. 

\textbf{Fashion-MNIST\footnote{\url{https://github.com/zalandoresearch/fashion-mnist}.}}(FMNIST) consists of 60,000 training instances and 10,000 testing instances. Each is a 28×28 grayscale image, which is labeled to one of ten classes.  We follow the model architecture, Network in Network (NiN) ~\cite{min2014nin} used in~\cite{levine2021deep} as the sub-classifier architecture for ensemble. 

\textbf{CIFAR-10\footnote{\url{https://www.cs.toronto.edu/~kriz/cifar.html}.}} contains 60,000 images of size 32×32×3 pixels, 50,000 for training and 10,000 for testing. Each of the instances is labeled to one of ten classes. We follow ~\cite{levine2021deep} to use NiN with full data augmentation as the sub-classifier architecture for ensemble.

\section{More Experimental Results}

\subsection{More Ablation Studies}

\paragraph{Impact of Sub-Problem Scale $\Delta$}
\label{sec:delta}
Table~\ref{tab:delta} reports the impact of $\Delta$ on the collective robustness of hash bagging when the time is limited to $2$s per sample. The collective robustness is reported in the form of a percentage. Namely, $13.00 \pm 2.76$ means that, there are $13\%$ predictions are certifiably simultaneously robust in average, with the variance $2.76$, which is to compute over $6$ randomly selected $\Delta$-size sub-problems. We can empirically tell that when the poison budget $r$ is low, a large $\Delta$ might prevent us from computing the optimal objective value. When the poison budget $r$ is high, we can easily find an attack to corrupt a large portion of predictions for the small $\Delta$-size sub-testset, while finding a better solution for the large $\Delta$-size sub-problem at the meantime. As a result, the optimal $\Delta$ increases with the poison budget $r$ as shown in Table~\ref{tab:delta}.

\paragraph{Impact of Solving Time $t$}
\label{sec:t}
Fig.~\ref{fig:ablate_t} reports the impact of solving time $t$ on the certified collective robustness of hash bagging if we do not apply decomposition, on CIFAR-10. We observe that the collective robustness roughly increases linearly with $\log(t)$, which suggests that directly increasing the solving time is not an effective way to improve the certified collective robustness. 

\begin{figure}[tb!]
    \centering
    \includegraphics[width=0.8\linewidth]{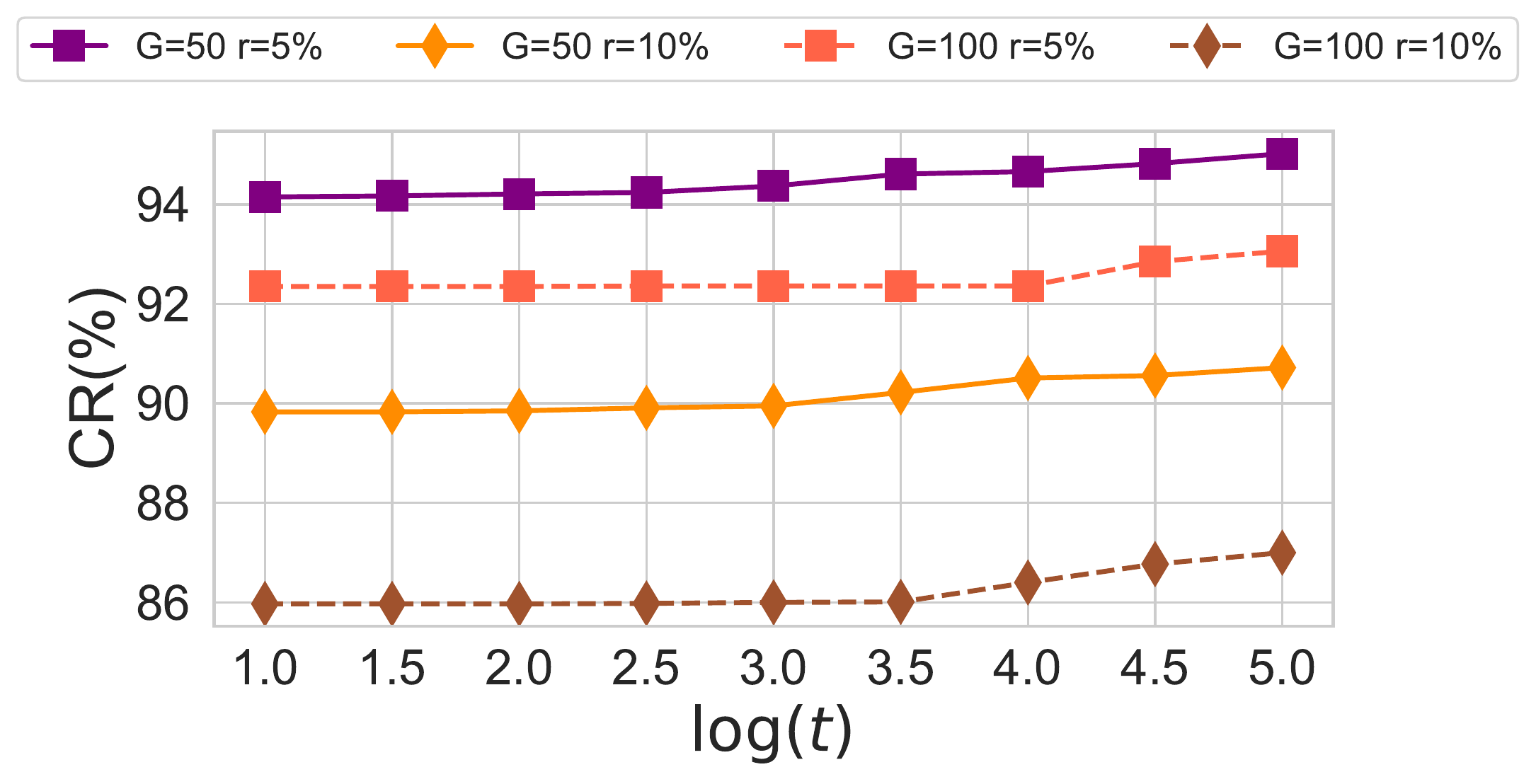}
    \caption{Impact of $t$ on CIFAR-10 ($K=N/G$).}
    \label{fig:ablate_t}
\end{figure}

\subsection{More Evaluation Results}
\label{sec:prob}
Table~\ref{tab:bank_more}, Table~\ref{tab:electricity_more}, Table~\ref{tab:fmnist_more}, Table~\ref{tab:cifar_more} report the detailed empirical results on Bank, Electricity, FMNIST, CIFAR-10, respectively. Specifically, we also compare to the probabilistic certification method~\cite{jia2021intrinsic}, where the confidence is set to be $0.999$ (the official implementation), and the number of sub-classifiers is set to be the same number used in the other certifications for the computational fairness. Note that the probabilistic certification cannot be applied to hash bagging, because it assumes that the sub-trainsets are randomly subsampled (with replacement) from the trainset. The empirical results demonstrate that, collective certification $>$ sample-wise certification $>$ probabilistic certification in terms of the certified collective robustness and the certified accuracy, on vanilla bagging. We observe that probabilistic certification performs poorly when $G$ is small, because the confidence interval estimation in probabilistic certification highly relies on the number of sub-classifiers.

\begin{table*}
\centering
\caption{(Bank: $M=10,000$; $K=5\%N$) Comparison on the certified collective robustness and the certified accuracy at $r=5\%, \ldots, 50\%$ ($\times G$), where $r= r_{\rm{ins}}+r_{\rm{del}}+2 r_{\rm{mod}}$ refers to the poison budget. \textbf{Sample-wise} and \textbf{Collective} refer to sample-wise and collective certification respectively. \textbf{Probabilistic} refers to the probabilistic certification proposed in~\cite{jia2021intrinsic}. \textbf{CR} and \textbf{CA} refer to the certified collective robustness and the certified accuracy respectively. \down{\alpha} denotes the relative gap between $M_{\rm{ATK}}$ guaranteed by the collective certification and $M_{\rm{ATK}}$ of the sample-wise certification. NaN: division by zero.} \label{tab:bank_more}
\begin{adjustbox}{width=\linewidth}
    \begin{tabular}{lll | lccccccccccc}
    \toprule
   G & Bagging & Certification & Metric & $5\%$ & $10\%$ & $15\%$ & $20\%$ & $25\%$ & $30\%$ & $35\%$ & $40\%$ & $45\%$ & $50\%$\\
    \midrule
   \multirow{24}{*}{20} 
   & \multirow{12}{*}{Vanilla} &  \multirow{4}{*}{Sample-wise} & CR & 3917 & 0 & 0 & 0 & 0 & 0 & 0 & 0 & 0 & 0 \\
    & & & $M_{\rm{ATK}}$ & 6083 & 10000 & 10000 & 10000 & 10000 & 10000 & 10000 & 10000 & 10000 & 10000 \\
    \cmidrule(l){4-14}
    & & &  CA & 3230 & 0 & 0 & 0 & 0 & 0 & 0 & 0 & 0 & 0 \\
    & & & $M_{\rm{ATK}}$ & 4790 & 8020 & 8020 & 8020 & 8020 & 8020 & 8020 & 8020 & 8020 & 8020\\
    \cmidrule(l){3-14}
    &  & \multirow{2}{*}{Probabilistic} & CR & 0 & 0 & 0 & 0 & 0 & 0 & 0 & 0 & 0 & 0\\
    \cmidrule(l){4-14}
    & & &  CA & 0 & 0 & 0 & 0 & 0 & 0 & 0 & 0 & 0 & 0\\
    \cmidrule(l){3-14}
    &  & \multirow{4}{*}{Collective} & CR & 4449 & 0 & 0 & 0 & 0 & 0 & 0 & 0 & 0 & 0\\
    & & & $M_{\rm{ATK}}$ & \down{8.74} & NaN & NaN & NaN & NaN & NaN & NaN & NaN & NaN & NaN \\
    \cmidrule(l){4-14}
    & & &  CA & 3588 & 0 & 0 & 0 & 0 & 0 & 0 & 0 & 0 & 0 \\
    & & & $M_{\rm{ATK}}$ & \down{7.47} & NaN & NaN & NaN & NaN & NaN & NaN & NaN & NaN & NaN\\
    \cmidrule(l){2-14}
    & \multirow{10}{*}{Hash} &  \multirow{4}{*}{Sample-wise} & CR & 9599 & 9009 & 7076 & 5778 & 4686 & 3772 & 2880 & 2157 & 1485 & 289 \\
    & & & $M_{\rm{ATK}}$ & 401 & 991 & 2924 & 4222 & 5314 & 6228 & 7120 & 7843 & 8515 & 9711\\
    \cmidrule(l){4-14}
    & & &  CA & 7788 & 7403 & 5755 & 4644 & 3817 & 3036 & 2283 & 1659 & 1106 & 284\\
    & & & $M_{\rm{ATK}}$ & 232 & 617 & 2265 & 3376 & 4203 & 4984 & 5737 & 6361 & 6914 & 7736 \\
    \cmidrule(l){3-14}
    &  & \multirow{4}{*}{Collective} & CR &  \textbf{9718} & \textbf{9209} & \textbf{7270} & \textbf{5968} & \textbf{4930} & \textbf{3915} & \textbf{3076} & \textbf{2294} & \textbf{1503} & \textbf{289}\\
    & & & $M_{\rm{ATK}}$ & \down{29.7} & \down{20.2} & \down{6.63} & \down{4.50} & \down{4.59} & \down{2.30} & \down{2.75} & \down{1.75} & \down{0.21} & \down{0.00}\\
    \cmidrule(l){4-14}
    & & &  CA & \textbf{7831} & \textbf{7464} & \textbf{5806} & \textbf{4685} & \textbf{3881} & \textbf{3091} & \textbf{2349} & \textbf{1689} & \textbf{1112} & \textbf{284}\\
    & & & $M_{\rm{ATK}}$ & \down{18.5} & \down{9.89} & \down{2.25} & \down{1.21} & \down{1.52} & \down{1.10} & \down{1.15} & \down{0.47} & \down{0.09} & \down{0.00}\\
    \hline
    \multirow{24}{*}{40} 
    & \multirow{14}{*}{Vanilla} &  \multirow{4}{*}{Sample-wise} & CR & 5250 & 1870 & 0 & 0 & 0 & 0 & 0 & 0 & 0 & 0\\
    & & & $M_{\rm{ATK}}$ & 4750 & 8130 & 10000 & 10000 & 10000 & 10000 & 10000 & 10000 & 10000 & 10000 \\
    \cmidrule(l){4-14}
    & & &  CA & 4160 & 1408 & 0 & 0 & 0 & 0 & 0 & 0 & 0 & 0 \\
    & & & $M_{\rm{ATK}}$ & 3913 & 6665 & 8073 & 8073 & 8073 & 8073 & 8073 & 8073 & 8073 & 8073\\
    \cmidrule(l){3-14}
    &  & \multirow{2}{*}{Probabilistic} & CR & 1509 & 1095 & 751 & 0 & 0 & 0 & 0 & 0 & 0 & 0\\
    \cmidrule(l){4-14}
    & & &  CA & 1049 & 705 & 407 & 0 & 0 & 0 & 0 & 0 & 0 & 0 \\
    \cmidrule(l){3-14}
    &  & \multirow{4}{*}{Collective} & CR & 5385 & 2166 & 0 & 0 & 0 & 0 & 0 & 0 & 0 & 0  \\
    & & & $M_{\rm{ATK}}$ & \down{2.84} & \down{3.64} & NaN & NaN & NaN & NaN & NaN & NaN & NaN & NaN \\
    \cmidrule(l){4-14}
    & & &  CA & 4190 & 1647 & 0 & 0 & 0 & 0 & 0 & 0 & 0 & 0 \\
    & & & $M_{\rm{ATK}}$ & \down{0.77} & \down{3.58} & NaN & NaN & NaN & NaN & NaN & NaN & NaN & NaN \\
    \cmidrule(l){2-14}
    & \multirow{10}{*}{Hash} &  \multirow{4}{*}{Sample-wise} & CR & 9638 & 9301 & 6401 & 5376 & 4626 & 4061 & 3398 & 2551 & 1497 & 115\\
    & & & $M_{\rm{ATK}}$ & 362 & 699 & 3599 & 4624 & 5374 & 5939 & 6602 & 7449 & 8503 & 9885\\
    \cmidrule(l){4-14}
    & & &  CA & 7881 & 7679	& 5198 & 4354 & 3718 & 3229 & 2693 & 1976 & 1037 & 114\\
    & & & $M_{\rm{ATK}}$ & 192 & 394 & 2875 & 3719 & 4355 & 4844 & 5380 & 6097 & 7036 & 7959\\
    \cmidrule(l){3-14}
    &  & \multirow{4}{*}{Collective} & CR &  \textbf{9762} & \textbf{9475} & \textbf{6603} & \textbf{5572} & \textbf{4796} & \textbf{4209} & \textbf{3562} & \textbf{2665} & \textbf{1523} & \textbf{115}\\
    & & & $M_{\rm{ATK}}$ & \down{34.2} & \down{24.9} & \down{5.61} & \down{4.24} & \down{3.16} &\down{2.49} & \down{2.48} & \down{1.53} & \down{0.30} & \down{0.00}\\
    \cmidrule(l){4-14}
    & & &  CA & \textbf{7914} & \textbf{7718} & \textbf{5236} & \textbf{4396}	& \textbf{3751} & \textbf{3257} &\textbf{2720} &\textbf{2010} &\textbf{1049} &\textbf{114}\\
    & & & $M_{\rm{ATK}}$ & \down{17.2} & \down{9.90} & \down{1.32} & \down{1.13} & \down{0.76} &\down{0.58} &\down{0.50} &\down{0.56} &\down{0.17} &\down{0.00}\\
    \bottomrule
    \end{tabular}
    \end{adjustbox}
\end{table*}

\begin{table*}
\centering
\caption{(Electricity: $M=10,000$; $K=5\%N$) Certified collective robustness and certified accuracy.} \label{tab:electricity_more}
\begin{adjustbox}{width=\linewidth}
    \begin{tabular}{lll | lccccccccccc}
    \toprule
   G & Bagging & Certification & Metric & $5\%$ & $10\%$ & $15\%$ & $20\%$ & $25\%$ & $30\%$ & $35\%$ & $40\%$ & $45\%$ & $50\%$\\
    \midrule
   \multirow{24}{*}{20} 
   & \multirow{12}{*}{Vanilla} &  \multirow{4}{*}{Sample-wise} & CR & 9230 & 0 & 0 & 0 & 0 & 0 & 0 & 0 & 0 & 0 \\
    & & & $M_{\rm{ATK}}$ & 770 & 10000 & 10000 & 10000 & 10000 & 10000 & 10000 & 10000 & 10000 & 10000 \\
    \cmidrule(l){4-14}
    & & &  CA & 7321 & 0 & 0 & 0 & 0 & 0 & 0 & 0 & 0 & 0 \\
    & & & $M_{\rm{ATK}}$ & 418 & 7739 & 7739 & 7739 & 7739 & 7739 & 7739 & 7739 & 7739 & 7739\\
    \cmidrule(l){3-14}
    &  & \multirow{2}{*}{Probabilistic} & CR & 0 & 0 & 0 & 0 & 0 & 0 & 0 & 0 & 0 & 0\\
    \cmidrule(l){4-14}
    & & &  CA & 0 & 0 & 0 & 0 & 0 & 0 & 0 & 0 & 0 & 0\\
    \cmidrule(l){3-14}
    &  & \multirow{4}{*}{Collective} & CR & 9348 & 0 & 0 & 0 & 0 & 0 & 0 & 0 & 0 & 0\\
    & & & $M_{\rm{ATK}}$ & \down{15.3} & NaN & NaN & NaN & NaN & NaN & NaN & NaN & NaN & NaN\\
    \cmidrule(l){4-14}
    & & &  CA & 7394 & 0 & 0 & 0 & 0 & 0 & 0 & 0 & 0 & 0\\
    & & & $M_{\rm{ATK}}$ & \down{17.5} & NaN & NaN & NaN & NaN & NaN & NaN & NaN & NaN & NaN\\
    \cmidrule(l){2-14}
    & \multirow{10}{*}{Hash} &  \multirow{4}{*}{Sample-wise} & CR & 9858 & 9738 & 9602 & 9461 & 9293 & 9121 & 8928 & 8656 & 8294 & 2597 \\
    & & & $M_{\rm{ATK}}$ & 142 & 262 & 398 & 539 & 707 & 879 & 1072 & 1344 & 1706 & 7403\\
    \cmidrule(l){4-14}
    & & &  CA & 7681 & 7621	& 7538	& 7462	& 7362 & 7266 & 7157 & 6998 & 6767 & 2198 \\
    & & & $M_{\rm{ATK}}$ & 58 & 118 & 201 & 277 & 377 & 473 & 582 & 741 & 972 & 5541\\
    \cmidrule(l){3-14}
    &  & \multirow{4}{*}{Collective} & CR &  \textbf{9915} & \textbf{9821} & \textbf{9726} & \textbf{9608} & \textbf{9402} & \textbf{9302} & \textbf{9122} & \textbf{8829} & \textbf{8449} & \textbf{2605}\\
    & & & $M_{\rm{ATK}}$ & \down{40.1} & \down{31.7} & \down{31.1} & \down{27.3} & \down{23.9} & \down{20.6} & \down{18.1} & \down{12.9} & \down{9.08} & \down{0.11}\\
    \cmidrule(l){4-14}
    & & &  CA & \textbf{7701} & \textbf{7663}	& \textbf{7608} & \textbf{7547} & \textbf{7458} & \textbf{7366} & \textbf{7265} & \textbf{7102} & \textbf{6856} & \textbf{2200}\\
    & & & $M_{\rm{ATK}}$ & \down{34.5} & \down{35.6} & \down{34.8} & \down{30.7} & \down{25.5} & \down{21.1} & \down{18.6} & \down{14.0} & \down{9.16} & \down{0.04}\\
    \hline
    \multirow{24}{*}{40} 
    & \multirow{14}{*}{Vanilla} &  \multirow{4}{*}{Sample-wise} & CR & 9482 & 8648 & 0 & 0 & 0 & 0 & 0 & 0 & 0 & 0  \\
    & & & $M_{\rm{ATK}}$ & 518 & 1352 & 10000 & 10000 & 10000 & 10000 & 10000 & 10000 & 10000 & 10000\\
    \cmidrule(l){4-14}
    & & &  CA & 7466 & 6986 & 0 & 0 & 0 & 0 & 0 & 0 & 0 & 0 \\
    & & & $M_{\rm{ATK}}$ & 284 & 764 & 7750 & 7750 & 7750 & 7750 & 7750 & 7750 & 7750 & 7750\\
    \cmidrule(l){3-14}
    &  & \multirow{2}{*}{Probabilistic} & CR & 8489 & 8248 & 7848 & 0 & 0 & 0 & 0 & 0 & 0 & 0\\
    \cmidrule(l){4-14}
    & & &  CA & 6892 & 6742 & 6506 & 0 & 0 & 0 & 0 & 0 & 0 & 0\\
    \cmidrule(l){3-14}
    &  & \multirow{4}{*}{Collective} & CR & 9566 & 8817 & 0 & 0 & 0 & 0 & 0 & 0 & 0 & 0 \\
    & & & $M_{\rm{ATK}}$ & \down{16.2} & \down{12.5} & NaN & NaN & NaN & NaN & NaN & NaN & NaN & NaN\\
    \cmidrule(l){4-14}
    & & &  CA & 7513 & 7086 & 0 & 0 & 0 & 0 & 0 & 0 & 0 & 0\\
    & & & $M_{\rm{ATK}}$ & \down{16.5} & \down{13.1} & NaN & NaN & NaN & NaN & NaN & NaN & NaN & NaN\\
    \cmidrule(l){2-14}
    & \multirow{10}{*}{Hash} &  \multirow{4}{*}{Sample-wise} & CR & 9873 & 9769 & 9636 & 9491 & 9366 & 9213 & 9022 & 8774 & 8434 & 2516 \\
    & & & $M_{\rm{ATK}}$ & 127 & 231 & 364 & 509 & 634 & 787 & 978 & 1226 & 1566 & 7484\\
    \cmidrule(l){4-14}
    & & &  CA & 7681 & 7625	& 7546 & 7459 & 7399 & 7316 & 7204 & 7065 & 6860 & 2142\\
    & & & $M_{\rm{ATK}}$ & 69 & 125 & 204 & 291 & 351 & 434 & 546 & 685 & 890 & 5608\\
    \cmidrule(l){3-14}
    &  & \multirow{4}{*}{Collective} & CR &  \textbf{9919} & \textbf{9842} & \textbf{9755} & \textbf{9601} & \textbf{9461} & \textbf{9312} & \textbf{9127} & \textbf{8883} & \textbf{8537} & \textbf{2524}\\
    & & & $M_{\rm{ATK}}$ & \down{36.2} & \down{31.6} & \down{32.7} & \down{21.6} & \down{15.0} &\down{12.6} & \down{10.7} & \down{8.89} & \down{6.58} & \down{0.11}\\
    \cmidrule(l){4-14}
    & & &  CA & \textbf{7700} & \textbf{7661}	& \textbf{7613}	& \textbf{7536}	& \textbf{7457} & \textbf{7378} & \textbf{7274} & \textbf{7140} & \textbf{6918} & \textbf{2145}\\
    & & & $M_{\rm{ATK}}$ & \down{27.5} & \down{28.8} & \down{32.8} & \down{26.5} & \down{16.5} & \down{14.3} & \down{12.8} & \down{10.9} & \down{6.52} & \down{0.05}\\
    \bottomrule
    \end{tabular}
    \end{adjustbox}
\end{table*}

\begin{table*}
\centering
\caption{(FMNIST: $M=10,000$; $K=N/G$) Certified collective robustness and certified accuracy. \textbf{Decomposition}: collective certification with decomposition.} \label{tab:fmnist_more}
\begin{adjustbox}{width=\linewidth}
    \begin{tabular}{lll | lccccccccccc}
    \toprule
   G & Bagging & Certification & Metric & $5\%$ & $10\%$ & $15\%$ & $20\%$ & $25\%$ & $30\%$ & $35\%$ & $40\%$ & $45\%$ & $50\%$\\
    \midrule
   \multirow{25}{*}{50} 
   & \multirow{12}{*}{Vanilla} &  \multirow{4}{*}{Sample-wise} & CR & 7432 & 0 & 0 & 0 & 0 & 0 & 0 & 0 & 0 & 0\\
    & & & $M_{\rm{ATK}}$ & 2568 & 10000 & 10000 & 10000 & 10000 & 10000 & 10000 & 10000 & 10000 & 10000 \\
    \cmidrule(l){4-14}
    & & &  CA & 7283 & 0 & 0 & 0 & 0 & 0 & 0 & 0 & 0 & 0\\
    & & & $M_{\rm{ATK}}$ & 1683 & 8966 & 8966 & 8966 & 8966 & 8966 & 8966 & 8966 & 8966 & 8966\\
    \cmidrule(l){3-14}
    &  & \multirow{2}{*}{Probabilistic} & CR & 6897 & 6633 & 5918 & 5214 & 0 & 0 & 0 & 0 & 0 & 0\\
    \cmidrule(l){4-14}
    & & &  CA & 6799 & 6557 & 5891 & 5201 & 0 & 0 & 0 & 0 & 0 & 0 \\
    \cmidrule(l){3-14}
    &  & \multirow{4}{*}{Collective} & CR & 7727 & 0 & 0 & 0 & 0 & 0 & 0 & 0 & 0 & 0 \\
    & & & $M_{\rm{ATK}}$ & \down{11.5} & NaN & NaN & NaN & NaN & NaN & NaN & NaN & NaN & NaN \\
    \cmidrule(l){4-14}
    & & &  CA & 7515 & 0 & 0 & 0 & 0 & 0 & 0 & 0 & 0 & 0\\
    & & & $M_{\rm{ATK}}$ & \down{13.8} & NaN & NaN & NaN & NaN & NaN & NaN & NaN & NaN & NaN\\
    \cmidrule(l){2-14}
    & \multirow{14}{*}{Hash} &  \multirow{4}{*}{Sample-wise} & CR & 9576 & 9307	& 8932 & 8671 & 8238 & 7929 & 7456 & 7051 & 6146 & 308 \\
    & & & $M_{\rm{ATK}}$ & 424 & 693 & 1068 & 1329 & 1762 & 2071 & 2544 & 2949 & 3854 & 9692\\
    \cmidrule(l){4-14}
    & & &  CA & 8768 & 8635	& 8408 & 8246 & 7943 & 7700 & 7295 & 6943 & 6107 & 308\\
    & & & $M_{\rm{ATK}}$ & 198 & 331 & 558 & 720 & 1023 & 1266 & 1671 & 2023 & 2859 & 8658\\
    \cmidrule(l){3-14}
    &  & \multirow{4}{*}{Collective} & CR & \textbf{9726}	& 9410 & 9024 & 8761 & 8329 & 8024 & 7525 & 7126 & 6277 & 329\\
    & & & $M_{\rm{ATK}}$ & \down{35.4} & \down{14.9} & \down{8.61} & \down{6.77} & \down{5.16} & \down{4.59} & \down{2.71} & \down{2.54} & \down{3.40} & \down{0.22}\\
    \cmidrule(l){4-14}
    & & &  CA & \textbf{8833} & \textbf{8719} & 8493 & 8327 & 8022 & 7780 & 7370 & 7020 & 6247 & 327\\
    & & & $M_{\rm{ATK}}$ & \down{32.8} & \down{25.4} & \down{15.2} & \down{11.2} & \down{7.72} & \down{6.32} & \down{4.49} & \down{3.81} & \down{4.90} & \down{0.22}\\
    \cmidrule(l){3-14}
    &  & \multirow{4}{*}{Decomposition} & CR & 9666 & \textbf{9472} & \textbf{9124} & \textbf{8887} & \textbf{8491} & \textbf{8196} & \textbf{7672} & \textbf{7287} &\textbf{6300} & \textbf{308}\\
    & & & $M_{\rm{ATK}}$ & \down{21.2} & \down{23.8} & \down{18.0} & \down{16.2} & \down{14.4} & \down{12.9} & \down{8.49} & \down{8.00} & \down{4.00} & \down{0.00}\\
    \cmidrule(l){4-14}
    & & &  CA & 8812 & 8716 & \textbf{8527} & \textbf{8385} & \textbf{8119} & \textbf{7892} & \textbf{7491} & \textbf{7150} & \textbf{6271} & \textbf{308}\\
    & & & $M_{\rm{ATK}}$ & \down{22.2} & \down{24.5} & \down{21.3} & \down{19.3} & \down{17.2} & \down{15.2} & \down{11.7} & \down{10.2} & \down{5.74} & \down{0.00}\\
    \hline
    \multirow{25}{*}{100} 
    & \multirow{14}{*}{Vanilla} &  \multirow{4}{*}{Sample-wise} & CR & 7548 & 0 & 0 & 0 & 0 & 0 & 0 & 0 & 0 & 0\\
    & & & $M_{\rm{ATK}}$ & 2452 & 10000 & 10000 & 10000 & 10000 & 10000 & 10000 & 10000 & 10000 & 10000\\
    \cmidrule(l){4-14}
    & & &  CA & 7321 & 0 & 0 & 0 & 0 & 0 & 0 & 0 & 0 & 0\\
    & & & $M_{\rm{ATK}}$ & 1443 & 8764 & 8764 & 8764 & 8764 & 8764 & 8764 & 8764 & 8764 & 8764\\
    \cmidrule(l){3-14}
    &  & \multirow{2}{*}{Probabilistic} & CR & 7169 & 6808 & 6518 & 6187 & 5805 & 5395 & 4876 & 3791 & 0 & 0\\
    \cmidrule(l){4-14}
    & & &  CA & 6958 & 6660 & 6405 & 6103 & 5746 & 5363 & 4855 & 3787 & 0 & 0\\
    \cmidrule(l){3-14}
    &  & \multirow{4}{*}{Collective} & CR & 8053 & 0 & 0 & 0 & 0 & 0 & 0 & 0 & 0 & 0 \\
    & & & $M_{\rm{ATK}}$ & \down{20.6} & NaN & NaN & NaN & NaN & NaN & NaN & NaN & NaN & NaN\\
    \cmidrule(l){4-14}
    & & &  CA & 7746 & 0 & 0 & 0 & 0 & 0 & 0 & 0 & 0 & 0\\
    & & & $M_{\rm{ATK}}$ & \down{29.4} & NaN & NaN & NaN & NaN & NaN & NaN & NaN & NaN & NaN\\
    \cmidrule(l){2-14}
    & \multirow{14}{*}{Hash} &  \multirow{4}{*}{Sample-wise} & CR & 9538 & 9080 & 8653 & 8249 & 7823 & 7419 & 6928 & 6377 & 5611 & 147\\
    & & & $M_{\rm{ATK}}$ & 462 & 920 & 1347 & 1751 & 2177 & 2581 & 3072 & 3623 & 4389 & 9853\\
    \cmidrule(l){4-14}
    & & &  CA & 8554 & 8316 & 8049 & 7797 & 7486 & 7173 & 6759 & 6279 & 5568 & 147\\
    & & & $M_{\rm{ATK}}$ & 210 & 448 & 715 & 967 & 1278 & 1591 & 2005 & 2485 & 3196 & 8617 \\
    \cmidrule(l){3-14}
    &  & \multirow{4}{*}{Collective} & CR & 9611 & 9167 & 8754	& 8344 & 7912 & 7483 & 6980 & 6405 & 5631 & 147\\
    & & & $M_{\rm{ATK}}$ & \down{15.8} & \down{9.46} & \down{7.50} & \down{5.42} & \down{4.09} &\down{2.48} & \down{1.69} & \down{0.77} & \down{0.46} & \down{0.00}\\
    \cmidrule(l){4-14}
    & & &  CA & \textbf{8610}	& 8375 & 8116 & 7857 & 7558 & 7242 & 6830 & 6323 & 5628 & 147\\
    & & & $M_{\rm{ATK}}$ & \down{26.7} & \down{13.2} & \down{9.37} & \down{6.20} & \down{5.63} & \down{4.34} & \down{3.54} & \down{1.77} & \down{1.88} & \down{0.00}\\
    \cmidrule(l){3-14}
    &  & \multirow{4}{*}{Decomposition} & CR & \textbf{9631} & \textbf{9232} & \textbf{8837} & \textbf{8450} & \textbf{8036} & \textbf{7617} & \textbf{7104} & \textbf{6513} & \textbf{5726} & \textbf{147}\\
    & & & $M_{\rm{ATK}}$ & \down{20.1} & \down{16.5} & \down{13.6} & \down{11.5} & \down{9.78} & \down{7.67} & \down{5.73} & \down{3.75} & \down{2.62} & \down{0.00}\\
    \cmidrule(l){4-14}
    & & &  CA & 8595 & \textbf{8407} & \textbf{8152} & \textbf{7917} & \textbf{7639} & \textbf{7334} & \textbf{6897} & \textbf{6404} & \textbf{5676} & \textbf{147}\\
    & & & $M_{\rm{ATK}}$ & \down{19.5} & \down{20.3} & \down{14.4} & \down{12.4} & \down{12.0} & \down{10.1} & \down{6.88} & \down{5.03} & \down{3.38} & \down{0.00}\\
    \bottomrule
    \end{tabular}
    \end{adjustbox}
\end{table*}

\begin{table*}
\centering
\caption{(CIFAR-10: $M=10,000$; $K=N/G$) Certified collective robustness and certified accuracy.} \label{tab:cifar_more}
\begin{adjustbox}{width=\linewidth}
    \begin{tabular}{lll | lccccccccccc}
    \toprule
   G & Bagging & Certification & Metric & $5\%$ & $10\%$ & $15\%$ & $20\%$ & $25\%$ & $30\%$ & $35\%$ & $40\%$ & $45\%$ & $50\%$\\
    \midrule
   \multirow{25}{*}{50} 
   & \multirow{12}{*}{Vanilla} &  \multirow{4}{*}{Sample-wise} & CR & 2737 & 0 & 0 & 0 & 0 & 0 & 0 & 0 & 0 & 0\\
    & & & $M_{\rm{ATK}}$ & 7263 & 10000 & 10000 & 10000 & 10000 & 10000 & 10000 & 10000 & 10000 & 10000 \\
    \cmidrule(l){4-14}
    & & &  CA & 2621 & 0 & 0 & 0 & 0 & 0 & 0 & 0 & 0 & 0\\
    & & & $M_{\rm{ATK}}$ & 4375 & 6996 & 6996 & 6996 & 6996 & 6996 & 6996 & 6996 & 6996 & 6996\\
    \cmidrule(l){3-14}
    &  & \multirow{2}{*}{Probabilistic} & CR & 1820 & 1529 & 876 & 490 & 0 & 0 & 0 & 0 & 0 & 0\\
    \cmidrule(l){4-14}
    & & &  CA & 1781 & 1501 & 867 & 488 & 0 & 0 & 0 & 0 & 0 & 0 \\
    \cmidrule(l){3-14}
    &  & \multirow{4}{*}{Collective} & CR & 3621 & 0 & 0 & 0 & 0 & 0 & 0 & 0 & 0 & 0 \\
    & & & $M_{\rm{ATK}}$ & \down{12.2} & NaN & NaN & NaN & NaN & NaN & NaN & NaN & NaN & NaN \\
    \cmidrule(l){4-14}
    & & &  CA & 3335 & 0 & 0 & 0 & 0 & 0 & 0 & 0 & 0 & 0\\
    & & & $M_{\rm{ATK}}$ & \down{16.3} & NaN & NaN & NaN & NaN & NaN & NaN & NaN & NaN & NaN\\
    \cmidrule(l){2-14}
    & \multirow{14}{*}{Hash} &  \multirow{4}{*}{Sample-wise} & CR & 8221 & 7268 & 6067 & 5320 & 4229 & 3573 & 2635 & 2019 & 978 & 39 \\
    & & & $M_{\rm{ATK}}$ & 1779 & 2732 & 3933 & 4680 & 5771 & 6427 & 7365 & 7981 & 9022 & 9961 \\
    \cmidrule(l){4-14}
    & & &  CA & 6305 & 5864 & 5186 & 4705 & 3884 & 3339 & 2520 & 1961 & 962 & 39\\
    & & & $M_{\rm{ATK}}$ & 691 & 1132 & 1810 & 2291 & 3112 & 3657 & 4476 & 5035 & 6034 & 6957 \\
    \cmidrule(l){3-14}
    &  & \multirow{4}{*}{Collective} & CR & 8393 & 7428 & 6204	& 5435 & 4290 & 3624 & 2664 & 2043 & 1034 & 40\\
    & & & $M_{\rm{ATK}}$ & \down{9.67} & \down{5.86} & \down{3.48} & \down{2.46} & \down{1.06} & \down{0.79} & \down{0.39} & \down{0.30} & \down{0.62} & \down{0.01}\\
    \cmidrule(l){4-14}
    & & &  CA & 6410	& 5985 & 5342	& 4848 & 4006 & 3434 & 2582 & 2007 & 1037 & 39\\
    & & & $M_{\rm{ATK}}$& \down{15.2} & \down{10.7} & \down{8.62} & \down{6.24} & \down{3.92} & \down{2.60} & \down{1.38} & \down{0.91} & \down{1.24} & \down{0.00}\\
    \cmidrule(l){3-14}
    &  & \multirow{4}{*}{Decomposition} & CR & \textbf{8694} & \textbf{7854} & \textbf{6686} & \textbf{5912} & \textbf{4826} & \textbf{4067} & \textbf{2995} & \textbf{2277} & \textbf{996} & \textbf{39}\\
    & & & $M_{\rm{ATK}}$ & \down{26.6} & \down{21.4} & \down{15.7} & \down{12.6} & \down{10.3} & \down{7.69} & \down{4.89} & \down{3.23} & \down{0.20} & \down{0.00}\\
    \cmidrule(l){4-14}
    & & &  CA & \textbf{6490} & \textbf{6147} & \textbf{5553} & \textbf{5113} & \textbf{4341} & \textbf{3733} & \textbf{2841} & \textbf{2234} & \textbf{1016} & \textbf{39}\\
    & & & $M_{\rm{ATK}}$ & \down{26.8} & \down{25.0} & \down{20.2} & \down{17.8} & \down{14.7} & \down{10.8} & \down{7.17} & \down{5.42} & \down{0.90} & \down{0.00}\\
    \hline
    \multirow{25}{*}{100} 
    & \multirow{14}{*}{Vanilla} &  \multirow{4}{*}{Sample-wise} & CR & 2621 & 0 & 0 & 0 & 0 & 0 & 0 & 0 & 0 & 0\\
    & & & $M_{\rm{ATK}}$ & 7379 & 10000 & 10000 & 10000 & 10000 & 10000 & 10000 & 10000 & 10000 & 10000\\
    \cmidrule(l){4-14}
    & & &  CA & 1876 & 0 & 0 & 0 & 0 & 0 & 0 & 0 & 0 & 0\\
    & & & $M_{\rm{ATK}}$ & 4378 & 6254 & 6254 & 6254 & 6254 & 6254 & 6254 & 6254 & 6254 & 6254\\
    \cmidrule(l){3-14}
    &  & \multirow{2}{*}{Probabilistic} & CR & 1473 & 1092 & 815 & 581 & 368 & 236 & 128 & 29 & 0 & 0\\
    \cmidrule(l){4-14}
    & & &  CA & 1395 & 1050 & 794 & 567 & 364 & 233 & 127 & 29 & 0 & 0\\
    \cmidrule(l){3-14}
    &  & \multirow{4}{*}{Collective} & CR & 2657 & 0 & 0 & 0 & 0 & 0 & 0 & 0 & 0 & 0 \\
    & & & $M_{\rm{ATK}}$ & \down{7.93} & NaN & NaN & NaN & NaN & NaN & NaN & NaN & NaN & NaN\\
    \cmidrule(l){4-14}
    & & &  CA & 2394 & 0 & 0 & 0 & 0 & 0 & 0 & 0 & 0 & 0\\
    & & & $M_{\rm{ATK}}$ & \down{11.8} & NaN & NaN & NaN & NaN & NaN & NaN & NaN & NaN & NaN\\
    \cmidrule(l){2-14}
    & \multirow{14}{*}{Hash} &  \multirow{4}{*}{Sample-wise} & CR & 7685 & 5962 & 4612 & 3504 & 2593 & 1833 & 1217 & 658 & 222 & 1\\
    & & & $M_{\rm{ATK}}$ & 2315 & 4038 & 5388 & 6496 & 7407 & 8167 & 8783 & 9342 & 9778 & 9999 \\
    \cmidrule(l){4-14}
    & & &  CA & 5396	& 4571 & 3787 & 3008 & 2315 & 1694 & 1166 & 634 & 218 & 1\\
    & & & $M_{\rm{ATK}}$ & 858 & 1683 & 2467 & 3246 & 3939 & 4560 & 5088 & 5620 & 6036 & 6253 \\
    \cmidrule(l){3-14}
    &  & \multirow{4}{*}{Collective} & CR & 7744 & 5974 & 4618	& 3509 & 2598 & 1838 & 1221 & 660 & 224 & 1 \\
    & & & $M_{\rm{ATK}}$ & \down{2.54} & \down{0.30} & \down{0.11} & \down{0.08} & \down{0.07} & \down{0.06} & \down{0.05} & \down{0.02} & \down{0.02} & \down{0.00}\\
    \cmidrule(l){4-14}
    & & &  CA & 5475 & 4650	& 3825	& 3030	& 2330 & 1710 & 1174 & 638 & 224 & 1\\
    & & & $M_{\rm{ATK}}$ & \down{9.21} & \down{4.69} & \down{1.54} & \down{0.68} & \down{0.38} & \down{0.35} & \down{0.16} & \down{0.07} & \down{0.10} & \down{0.00}\\
    \cmidrule(l){3-14}
    &  & \multirow{4}{*}{Decomposition} & CR & \textbf{8137} & \textbf{6469} & \textbf{5061} & \textbf{4035} & \textbf{2987} & \textbf{2032} & \textbf{1341} & \textbf{691} & \textbf{222} & \textbf{1}\\
    & & & $M_{\rm{ATK}}$ & \down{19.5} & \down{12.5} & \down{8.33} & \down{8.17} & \down{5.32} & \down{2.44} & \down{1.41} & \down{0.35} & \down{0.00} & \down{0.00}\\
    \cmidrule(l){4-14}
    & & &  CA & \textbf{5570} & \textbf{4841} & \textbf{4098} & \textbf{3338} & \textbf{2635} & \textbf{1928} & \textbf{1273} & \textbf{704} & \textbf{218} & \textbf{1}\\
    & & & $M_{\rm{ATK}}$ & \down{20.3} & \down{16.0} & \down{12.6} & \down{10.2} & \down{8.12} & \down{5.13} & \down{2.10} & \down{1.25} & \down{0.00} & \down{0.00}\\
    \bottomrule
    \end{tabular}
    \end{adjustbox}
\end{table*}

\section{Limitations}
As a defense against data poisoning, the main limitation of bagging is that we need to train multiple sub-classifiers to achieve a high certified robustness, because bagging actually exploits the majority voting based redundancy to trade for the robustness. Moreover, our collective certification does not take into account any property of the sub-classifiers, because our certification is agnostic towards the classifier architectures. Therefore, if we can specify the model architecture, we can further improve the certified robustness by exploiting the intrinsic property of the base model. Our collective certification needs to solve a costly NP-hard problem. A future direction is to find a collective robustness lower bound in a more effective way. 


\end{document}